\documentclass[10pt]{article}
\usepackage{fullpage,epsfig,graphics,cancel,bbm,slashed}
\usepackage{amsthm,amssymb,amsmath,amsfonts,amscd,amsbsy}
\usepackage{lmodern}
\usepackage{slantsc}
\usepackage{psfrag,hyperref,color}
\usepackage[]{youngtab}
\usepackage{graphicx}
\usepackage{wasysym}
\usepackage{mathdots}
\usepackage{mathrsfs}

\usepackage{tikz}
\usetikzlibrary{arrows,chains,matrix,positioning,scopes,snakes,decorations.pathmorphing,bending}
\usepackage{tikz-cd}
\tikzset{snake it/.style={decorate, decoration=snake}}

\makeatletter
\tikzset{join/.code=\tikzset{after node path={%
\ifx\tikzchainprevious\pgfutil@empty\else(\tikzchainprevious)%
edge[every join]#1(\tikzchaincurrent)\fi}}}
\makeatother

\tikzset{>=stealth',every on chain/.append style={join},
         every join/.style={->}}
\tikzset{
    >=stealth',
    punkt/.style={
           rectangle,
           rounded corners,
           draw=black, very thick,
           text width=6.5em,
           minimum height=2em,
           text centered},
    pil/.style={
           ->,
           thick,
           shorten <=2pt,
           shorten >=2pt,}
}

\newcommand{\BB}{\mathbb}
\newcommand{\SF}{\mathsf}
\newcommand{\FR}{\mathfrak}

\def\B{\SF{B}}
\def\P{\SF{P}}

\newcommand{\bea}{\begin{eqnarray}}
\newcommand{\eea}{\end{eqnarray}}
\newcommand{\nn}{\nonumber}
\newcommand{\Tr}{\operatorname{Tr}}

\newcommand{\sbullet}{\vcenter{\hbox{\tiny$\bullet$}}}
\newcommand{\udl}{\underline}
\newcommand{\bra}{\langle}
\newcommand{\ket}{\rangle}
\newcommand{\im}{\operatorname{Im}}
\newcommand{\re}{\operatorname{Re}}
\newcommand{\To}{\Rightarrow}

\newcommand{\ext}{\operatorname{Ext}}

\newcommand{\opn}{\operatorname}

\def\ot{\leftarrow}
\def\HH{H\hspace{-.3em}H}
\newcommand\hcancel[2][black]{\setbox0=\hbox{$#2$}%
\rlap{\raisebox{.45\ht0}{\textcolor{#1}{\rule{\wd0}{1pt}}}}#2}

\def\Gc{\Gamma}
\def\gd{\delta}
\def\Gd{\Delta}
\def\ep{\epsilon}

\def\gs{\sigma}

\def\gl{\lambda}
\def\Gl{\Lambda}
\def\Go{\Omega}
\def\go{\omega}

\DeclareMathAlphabet{\mathpzc}{OT1}{pzc}{m}{it}

\newtheorem{theorem}{Theorem}[section]
\newtheorem{lemma}[theorem]{Lemma}
\newtheorem{proposition}[theorem]{Proposition}

\theoremstyle{definition}
\newtheorem{example}[theorem]{Example}
\newtheorem{remark}[theorem]{Remark}
\newtheorem{corollary}[theorem]{Corollary}
\newtheorem{definition}[theorem]{Definition}

\numberwithin{equation}{section}
\begin{document}
\begin{flushright} \small
UUITP-05/23
 \end{flushright}
\smallskip
\begin{center} \Large
{\bf Quantisation via Branes and Minimal Resolution}
 \\[12mm] \normalsize
{\bf Jian Qiu${}^{a,b}$} \\[8mm]
 {\small\it
    ${}^a$Department of Physics and Astronomy, Uppsala University,\\
        \vspace{.3cm}
      ${}^b$ Mathematics Institute,  Uppsala University, \\
      \vspace{.3cm}
   Uppsala, Sweden\\  }
\end{center}
\vspace{7mm}
\begin{abstract}
 \noindent
 The `brane quantisation' is a quantisation procedure developed by Gukov and Witten \cite{Gukov:2008ve}.
 We implement this idea by combining it with the tilting theory and the minimal resolutions. This way, we can realistically compute the deformation quantisation on the space of observables acting on the Hilbert space. We apply this procedure to certain quantisation problems in the context of generalised K\"ahler structure on $\mathbb{P}^2$. Our approach differs from and complements that of Bischoff and Gualtieri \cite{Bischoff:2021ixy}.
 We also benefitted from an important technical tool: a combinatorial criterion for the Maurer-Cartan equation, developed in \cite{BarmeierWang} by Barmeier and Wang.
\end{abstract}

\eject
\normalsize

\tableofcontents
\section{Introduction}
Let $(M^{2n},\go)$ be a symplectic manifold, which can be e.g. the phase space of a mechanical system. By quantising $(M,\go)$, we mean that we construct a Hilbert space ${\cal H}$, and promote functions on $M$ into operators $f\to\hat f$ that act on ${\cal H}$. This is done in such a way that the Poisson bracket $\{f,g\}$ get promoted into commutators of operators $[\hat f,\hat g]$.

The Hilbert space can often be constructed with the geometrical quantisation, which works in a fairly broad set of scenarios. What one requires is that $[\go/(2\pi\hbar)]$ be of integral class, and then one has a line bundle ${\cal L}$ with curvature $F=i\go/(2\pi\hbar)$. The next step is to choose a polarisation, that is, a Lagrangian foliation of $M$. Put simply, this foliation designates half of the coordinates of $M$ as coordinates and the other half as momenta.
The Hilbert space ${\cal H}$ will then consist of sections of ${\cal L}$ that are covariantly constant along the leaves of the foliation.
This leads to sections of delta-function support along a selected set of leaves known as the `Bohr-Sommerfeld leaves'. As a result they are quite singular. When $M$ is K\"ahler, one can choose the K\"ahler polarisation, and set ${\cal H}$ to be the holomorphic sections of ${\cal L}$. If further $M$ is compact, such sections are of finite dimension and so will be the Hilbert space. The most elegant example is perhaps the full flag manifolds $G/B$ (where $G$ is a compact Lie group and $B$ its Borel subgroup). The integrality of $[\go/2\pi\hbar]$ amounts to choosing a dominant weight and the resulting Hilbert space gives the various highest weight representations of $G$. This is the content of the celebrated Borel-Weil-Bott theorem.

It is however much harder to promote functions on $M$ into operators, mainly due to the normal ordering ambiguity. Only in very simple cases and for restricted class of functions do we have a procedure to assign operators to functions. The brane quantisation procedure is meant to provide a framework to solve this very problem. The main construction is to embed $M$ into a manifold of double the dimension $(M^{2n},\go)\hookrightarrow (X^{4n},\Go)$, with $\Go$ a holomorphic symplectic form and $\re\Go|_M=\go$. This is denoted as `complexifying $M$', which is to be understood as saying that the functions on $M$ we want to quantise should have an analytic continuation into holomorphic functions on $X$.

The proposal in \cite{Gukov:2008ve} is to study the $A$-model on $X$, with symplectic form $\im\Go$, which deforms the multiplication rule of the holomorphic functions on $X$. However, in addition to the Lagrangian branes in $A$-model, it was necessary to consider also the coisotropic branes \cite{Kapustin:2001ij}, which is much less studied than the Lagrangian ones. There was an appealing proposal of the hom space between these branes in \cite{Bischoff:2021ixy}, see also \cite{leung2023hidden} for some recent related work.
But we put aside the $A$-model in this work and use the deformation quantisation to quantise the holomorphic functions. Important distinction is to be made between deformation quantisation and quantisation: the former typically does not give an actual deformation as the multiplication ends up being a formal power series in $\hbar$. Even when the series does converge, there will be no integrality condition for $\hbar$. But as we shall see in the examples of this paper, we do get integrality conditions out of deformation quantisation.

We said earlier that the real functions to be quantised are now continued to holomorphic functions on $X$. This is a huge advantage, since we have far more handles on the holomorphic functions than the real ones. Furthermore, rather than deforming the structure sheaf ${\cal O}_X$, we choose to deform the entire category of coherent sheaves on $X$. By doing more than what is asked here, we are actually doing less. This is because in many cases we can use the tilting theory to give an economic $A_{\infty}$-category model for $D^b(X)$. The setting here is that of \cite{lowen_2008}.
The reader may see in sec.\ref{sec_D4} that we get such a deformation effortlessly by going via the derived category and yet it translates to a very complicated deformation of the original algebra ${\cal O}_X$.

The rest of the paper follows this line of thought. Let us highlight one main result. Take a 3-dim vector space $V$, and set $X=T^*\BB{P}(V)$ (or put simply $T^*\BB{P}^2)$. Its holomorphic functions are $H^0(X):=H^0(X,{\cal O}_X)$, and are generated by $H^0(\BB{P}(V),T\BB{P}(V))=(V\otimes V^*)_{tf}$ ($tf$ stands for trace free). It is this algebra we will deform. The space $X$ has a holomorphic symplectic form $\Go_s$ depending on one parameter $s$. In particular, when $s=1$, then $(X,\Go_1)$ is a holomorphic symplectic groupoid over $\BB{P}(V)$ featuring in \cite{Bischoff:2018kzk}). From the deformation quantisation of $(X,\Go_s)$, and the Kodaira-Spencer deformation, we get a three parameter deformation of the algebra of holomorphic functions on $X$. Concretely we denote the basis of $V\otimes V^*$ as $X_i^j$, $i,j=1,2,3$ with $\Tr[X]=\sum_pX_p^p=0$. We obtain a deformed algebra
\bea
&& X_p^p=t+\frac12s,~~~X_i^p\star X_p^l=(t-\frac32s)X_i^l,\nn\\
&&X_i^j\star X_k^l=q^{j<k}q^{l<j}q^{k<l}X_i^l\star X_k^j-s(\gd_k^jX_i^l-q^{j<k}q^{l<j}X_i^j\gd_k^l)\nn\\
&&X_i^j\star X_k^l=q^{j<k}q^{k<i}q^{i<j}X_k^j\star X_i^l-s(\gd_k^jX_i^l-q^{j<k}q^{k<i}\gd_i^jX_k^l),\nn\\
&&X_i^j\star X_k^l=q^{j<k}q^{k<i}q^{l<j}q^{i<l}X_k^l\star X_i^j-s(\gd_k^jX_i^l-q^{j<k}q^{k<i}q^{l<j}\gd_i^lX_k^j).\label{deformed_TCP2}\eea
The 2nd and 3rd line imply the 4th line. Further $q^{i<j}$ is defined as
\bea \begin{array}{|c|c|c|c|c|c|c|c|}
\hline
 i,j & 12 & 23 & 31 & 21 & 32 & 13 & i=j \\
\hline
q^{i<j} & q & q & q & q^{-1} & q^{-1} & q^{-1} & 1 \\
\hline\end{array}.\nn\eea
For example
\bea
&&X_1^p\star X_p^1=(t-\frac32s)X_1^1,~~~X_p^1\star X_1^p=(t+\frac32s)X_1^1-s(t+\frac12s)\nn\\
&&X_1^2\star X_2^3=q^{-3}X_2^3\star X_1^2-sX_1^3,\nn\\
&&X_1^2\star X_2^1=X_2^2\star X_1^1-sX_1^1,\nn\\
&&X_2^1\star X_3^2=q^{-3}X_3^2\star X_2^1+sq^{-3}X_3^1.\nn\eea
When $s=1$, and $(X,\Go_1)$ is a symplectic groupoid, so we expect a co-multiplication and an antipode map on $H^0(X)$. This will be developed in a future work.

Finally, a word about the scope of applicability of our approach. The key, as will be explained in sec.\ref{sec_DQ}, is the tilting object that helps to give a small model of the category to be deformed. Therefore one might hope to apply this approach to some conic toric Calabi-Yau's, whose tilting object is related to the dimer model, or certain Gorenstein singularities, whose singularity category is equivalent to the stable category of matrix factorisation category.

\bigskip

{\bf Acknowledgements:} It is my pleasure to thank Severin Barmeier for patiently answering many of my technical questions. I'd like to give a shout-out to his work with Zhengfang Wang \cite{BarmeierWang}, without which the current work would be impossible. I also thank M. Gualtieri for discussion and for sharing his idea about brane quantisation with me.

\section{The setting}\label{sec_S}
Let $(M,\go)$ be a K\"ahler manifold. We assume that we have an embedding $(M,\go)\to (X,\Go)$ with $\Go$ a holomorphic symplectic form such that $\re\Go|_M=\go$ and $\im \Go|_M=0$. One should think of $\Go$ as the analytic continuation of $\go$.
One then considers the $A$-model on $X$ with symplectic form $\im \Go$, for which $M$ is a Lagrangian brane. One has to consider also the coisotropic branes. We will only need the space filling coisotropic branes, consisting of the data of a line bundle ${\cal L}$ with curvature $iF$, such that
\bea J=(\im\Go)^{-1}F\nn\eea
considered as a map $TX\to TX$ satisfies $J^2=-1$. Amongst these we will consider the canonical coisotropic branes $B_{cc}$, for which the line bundle is trivial ${\cal L}\sim{\cal O}_X$ (and $F$ is exact).

The expectation is that the open strings starting and ending on $B_{cc}$ will have Hilbert space $H^{\sbullet}(X,{\cal O}_X)$, while those having one end on $B_{cc}$ and one end on $M$ will have Hilbert space $H^{\sbullet}(X,E)$, where $E$ is a vector bundle or a sheaf that one attaches to $M$. The holomorphic functions act naturally on $H^{\sbullet}(X,E)$ as concatenation of strings, as caricatured in fig.\ref{fig_string_act}.
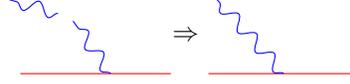
\begin{figure}[h]
\begin{center}
\begin{tikzpicture}[scale=1]
\draw [red,-] (-1.2,0) -- (0.8,0);
\path [draw=blue,snake it]
    (0,0) -- (-.5,0.7);
\path [draw=blue,snake it, bend right=15]
    (-.7,0.8) -- (-1.4,1);
\node at (1,.5) {$\To$};
\draw [red,-] (1.3,0) -- (3.3,0);
\path [draw=blue,snake it, bend right=25]
    (2.3,0) -- (1.4,1);
\end{tikzpicture}\caption{Concatenating strings gives the action of ${\cal O}_X$ on $E$.}\label{fig_string_act}
\end{center}
\end{figure}
Much of the above is what people would expect from string theory, but for the discussion below, one will not need to know the inner workings of string theory to understand the framework of quantisation via branes. Instead,
\begin{center}
\emph{In this work, strings are synonymous to a resolution of an algebra $A$ as a projective $A$-$A$ bi-module.}
\end{center}
Here the algebra $A$ can be a commutative algebra, a sheaf of algebra such as ${\cal O}_X$, or a quiver algebra. Encapsulated in this statement is the idea of path space and resolution of the diagonal. We will gradually unpack these as we go on.

We now want to deform the algebra structure ${\cal O}_X$. As said earlier, we do not go the route of $A$-model as we do not have the wherewithal to do such computations. Instead we will take up deformation quantisation as an alternative \cite{kontsevich2001deformation}. First, the infinitesimal deformations are controlled by the Hochschild cohomology group $\HH^2(X)$, which by the Hochschild-Kostant-Rosenberg (HKR) theorem \cite{HKR} is isomorphic to
\bea \HH^2(X)=H^2(X,{\cal O}_X)\oplus H^1(X,TX)\oplus H^0(X,\wedge^2TX).\label{HKR}\eea
Strictly speaking HKR proved the statement for smooth affine varieties, and so only the last summand above is present. The other two summands follow from taking \v{C}ech cover or a Hodge spectral sequence argument, see \cite{SWAN199657}.
The second summand in \eqref{HKR} controls the infinitesimal Kodaira-Spencer type deformation, while
the last summand $H^0(X,\wedge^2TX)$ represents the holomorphic bi-vector fields, wherein resides the holomorphic Poisson structures. We focus on these structures for now. Let $\Pi$ be a holomorphic Poisson tensor, it leads to the first order deformation to the commutative product on ${\cal O}_X$
\bea f\star g=fg+\hbar\{f,g\}_{\Pi}+\cdots.\nn\eea
To obtain the complete higher order deformations, one needs to solve the so called Maurer-Cartan (MC) equation, which we will recall in later sections. But ${\cal O}_X$ is too big an algebra and $\Pi$, as a representative of $\HH^2$, is too unwieldy for us to solve the MC equation.

Instead, let us deal with a wider problem of deforming the category of coherent sheaves. For this, we shall adopt a smaller model for $D^b(X)$ by using the tilting theory. Briefly, the tilting theory gives a derived equivalence $D^b(X)\simeq D^b(\opn{mod}_{\Gl})$ for some algebra $\Gl$ that is much smaller. To obtain this equivalence, consider the product $X\times X$ with $\Gd$ being the diagonal. Let $p,q$ be the projection to the first and second $X$, then the following is the identity functor in $D^b(X)$
\bea F\to Rp_*(q^*F\otimes^L{\cal O}_{\Gd}),~~~F\in D^b(X).\label{tilting}\eea
Put simply \eqref{tilting} is the derived category version of: $\int\gd(x-y)f(y)dy=f(x)$, where ${\cal O}_{\Gd}$ plays the role of the delta function and the pushdown $Rp_*$ replaces the integration over $y$.
But to compute the derived tensor with ${\cal O}_{\Gd}$, one needs to resolve ${\cal O}_{\Gd}$ as an ${\cal O}_{X\times X}$-module, this is where our earlier statement about path space and bi-module resolution comes in: heuristically, rather than imposing $x=y$ strictly as does $\gd(x-y)$, we impose $x=y$ up to homotopy by using a path to connect $x$ to $y$. The bar complex (see sec.\ref{sec_BR}) realises this idea of path space algebraically.
On the other hand, considering ${\cal O}_{\Gd}$ as an ${\cal O}_{X\times X}$ module amounts to considering ${\cal O}_X$ as an ${\cal O}_X$-bi-module, and the need to resolve ${\cal O}_X$ as an ${\cal O}_X$-bi-module brings us back to the bar complex.
Thus one can say that the strings provide us with a resolution geometrically while in practice one can work algebraically and pick any bi-module resolution one prefers. Optimally, such a resolution will be a complex of some vector bundles (they constitute the tilting object). The result of applying \eqref{tilting} then expresses any $F$ as extensions of the said vector bundles.

Assuming that the tilting procedure is done and we have the tilting object $T$. Then the algebra $\Gl=\opn{End}(T)$ is the algebra featuring in the equivalence $D^b(X)\simeq D^b(\opn{mod}_{\Gl})$, whereas the bar complex of $\Gl$ can be regarded as a leaner $A_{\infty}$-model for $D^b(X)$ (see remark \ref{rmk_A_infty_model}).
We will then be able to find a handier representative of $\Pi$ in $\HH^2$ and perform the deformation quantisation on $\Gl$, moreover, we will also incorporate the middle summand of \eqref{HKR}, i.e. the Kodaira-Spencer type deformation, in the same process.

\smallskip

Let us summarise the sort of examples for which we will implement the above procedure.
The recurring example in \cite{Gukov:2008ve} is $S^2$, presented as
\bea x^2+y^2+z^2=r^2\label{A_1_deform}\eea
equipped with the standard K\"ahler form $\go=dx\wedge dy/z$. By regarding $x,y,z$ as complex rather than real coordinates, the same equation cuts out a subvariety of $\BB{C}^3$, which is topologically $X=T^*S^2$. The same 2-form $dx\wedge dy/z$ is in fact the canonical holomorphic 2-form on $X$. The parameter $r^2$ controls the size of $\go=\re\Go$ and so is expected to have some integrality conditions.

As warmup examples, we will instead set $r^2=0$ above, and end up having an $A_1$-singularity $x^2+y^2+z^2=0$. For simplicity we make a change of variable and write the equation as $xy-z^2=0$.
The Auslander-Reiten theory for the algebra $R=\BB{C}[x,y,z]/(xy-z^2)$ is well-known, leading to the quiver (with relations)
\bea  \begin{tikzpicture}
  \matrix (m) [matrix of math nodes, row sep=.5em, column sep=5em]
    { e_1 & e_0 \\        };
  \path[->, font=\scriptsize, bend right = 5]
  (m-1-1) edge (m-1-2);
  \path[->, font=\scriptsize, bend left = 5]
  (m-1-1) edge (m-1-2);
  \path[->, blue, font=\scriptsize, bend right = 25]
  (m-1-2) edge (m-1-1);
  \path[->, blue,font=\scriptsize, bend left = 25]
  (m-1-2) edge (m-1-1);
  \end{tikzpicture}.\nn\eea
We will perform deformation quantisation on this algebra. The radius parameter $r^2$ in \eqref{A_1_deform} will be incorporated in the deformation quantisation as Kodaira-Spencer type deformation. Even better, we will obtain the same integrality condition for $r^2$, which is rather unusual for deformation quantisation.


Then we shall turn to the main application for our work. In the context of generalised K\"ahler structure on a manifold $M$, one naturally has a holomorphic Poisson structure $\pi$. Then to encode the generalised K\"ahler potential, one needs to integrate $\pi$ into a holomorphic symplectic groupoid \cite{Bischoff:2018kzk}, denoted as $(X,\Go)$. Inside $X$, we have $M$ embedded as the space of units, this is another natural setting in which the `complexification' $M\hookrightarrow X$ of Gukov and Witten makes sense.
One also needs to deform $\Go$ by a 2-form $F$: $\Go\to \Go+tF$, which serves to deform the complex structure on $X$ and eventually leads to the generalised K\"ahler structure on $M$.
We will apply the deformation quantisation to $(X,\Go)$ to deform ${\cal O}_X$ and also incorporate the deformation $F$ for the special case of $M=\BB{P}^2$. In the work \cite{Bischoff:2021ixy}, all toric surfaces are considered.
Moreover, the authors of loc. cit. also proposed a definition of the A-branes to compute the hom space of such branes.
In contrast, out approach stays on the B-side and the main tool is the more familiar homological algebra and deformation quantisation. We hope that our result would complement those of loc. cit.

\section{Bi-module Resolution and Deformation Quantisation}
For both the tilting theory and for deforming an algebra, we will need to resolve an algebra $A$ as a bi-module over itself.
In this section, we review the basics of the bar complex and some other related material. We assume that the reader is familiar with the $L_{\infty}$-structures.

We specify our setting first: apart from the familiar case where $A$ is a $k$-algebra, we will need to consider some quiver algebras, which arise as the endomorphism algebra of some tilting object. Let $Q$ be a quiver, $Q_0=\{e_i\}$ be the set of nodes and $Q_1$ the set of length 1 paths. Following the notation of \cite{BUTLER1999323}, let $S=k\bra e_i\ket$ be the algebra generated by the idempotents. Since $Q_1$ is an $S$-bi-module, we can consider the tensor algebra
\bea \Gc=TQ_1=\bigoplus_{n\geq0}\underbrace{Q_1\otimes_S\cdots\otimes_S Q_1}_n,\nn\eea
that is, $\Gc$ is the path algebra of $Q$. Let $R$ be the 2-sided ideal of relations and we let $\Gl=\Gc/R$.

\subsection{The Bar Resolution}\label{sec_BR}
To resolve $\Gl$ as a bi-module over itself, one can always use the bar (or Hochschild) complex $\B_{\sbullet}\to \Gl$. Though it is very inefficient, it always works and has many extra structures, as we shall recall.

At degree $n$, we have
\bea \B_n=\Gl\otimes \underbrace{\Gl[1]\otimes\cdots\otimes \Gl[1]}_n\otimes\Gl,~~~~\otimes:=\otimes_S.\nn\eea
We also use the bar notation and write $a_0[a_1|\cdots|a_n]a_{n+1}$ for an element of $\B_n$. The differential reads
\bea \partial a_0[a_1|\cdots|a_n]a_{n+1}&=&a_0a_1[a_2|\cdots|a_n]a_{n+1}-a_0[a_1a_2|a_3|\cdots|a_n]a_{n+1}\nn\\
&&+\cdots+(-1)^{n-1}a_0[a_1|\cdots|a_{n-2}|a_{n-1}a_n]a_{n+1}+(-1)^na_0[a_1|\cdots|a_{n-1}]a_na_{n+1}.\nn\eea
Note that as the tensor is over $S$, the $a_i$'s are composable paths.

The bar complex has a chain homotopy as a \emph{left} $\Gl$-module map: $k:\,\B_i\to \B_{i+1}$
\bea &k([a_1|\cdots|a_i]b)=(-1)^{i-1}[a_1|\cdots|a_i|b]\label{bar_k}\\
& \{k,\partial\}=1,\nn\eea
where if $b$ has length 0, then the $k$ is set as zero.  One can also define a right module version of $k$.

Let $M$ be, say, a left $\Gl$-module, then the tensor product
\bea \B_{\sbullet}\otimes_{\Gl}M\nn\eea
is a resolution of $M$. It is not unhelpful to regard $M$ as a brane and the $\B_{\sbullet}$ tensored to it as a path ending on the brane. The right-module version of $k$ above gives a contraction of the complex (as an $S$-module), which reflects the ability to shrink a path to its end point.

Let us denote the cochain complex as $\B^n:=\hom_{\Gl^e}(\B_n,\Gl)$, where the subscript ${}_{\Gl^e}$ means that the maps are bi-module maps.
Let $\gd$ denote the differential on $\B^{\sbullet}$ induced from $\partial$ on $\B_{\sbullet}$.
There is a natural element $m\in \B^2$, which maps $1[a|b]1$ to $ab$. Since the multiplication in $\Gl$ is associative, it is easy to check that
$\gd m=0$.

The cochain complex has the structure of a graded Lie algebra equipped with the Gerstenhaber bracket. To formulate the bracket, one first defines the composition $\circ_i$ between two cochains $f\in B^m$ and $g\in B^n$
\bea f\circ_ig([a_1|\cdots|a_{m+n-1}]):=f([a_1|\cdots|a_{i-1}|g([a_i|\cdots|a_{i+n-1}])|a_{i+n}|\cdots|a_{m+n-1}])\nn\eea
i.e. one plugs the outcome of $g$ to the $i^{th}$ slot of $f$. The Lie bracket reads
\bea [f,g]=f\circ g-(-1)^{(m-1)(n-1)}g\circ f,~~~~f\circ g=\sum_{i=1}^m(-1)^{(n-1)(i-1)}f\circ_ig.\label{Lie_bkt_bar}\eea
For example, for $f\in B^1$ and $g,h\in B^2$, then
\bea &&[f,g]=f\circ_1g-g\circ_1f-g\circ_2f,\nn\\
&&[g,h]=g\circ_1h-g\circ_2h+(g\leftrightarrow h).\nn\eea
From this definition it is also easy to check that $[m,m]=0$ and hence $[m,-]$ is nilpotent and so a differential. It is \emph{minus} the differential $\gd$: $[m,f]=-\gd f$.
\begin{remark}\label{rmk_A_infty_model}
  We stressed earlier that the tensor products here are over $S$, and so $\B$ consists of composable paths in $\Gl$.
  One can regard $\Gl$ as a category with the idempotents $e_i$ as objects and the paths as morphisms. Then $\B_n$ consists of $n$ composable morphisms and $\B^n=\hom_{\Gl^e}(\B_n,\Gl)$ gives a composition of the $n$ morphisms. We have $m\in\B^2$ and $[m,m]=0$ giving the associativity of compositions. But more generally let $g\in \B^{\sbullet}$ with $[g,g]=0$, then $g$ gives $\Gl$ an $A_{\infty}$-category structure, whose deformation is controlled by the Hochschild cohomology $\HH^{\sbullet}(\Gl)$ \cite{lowen_2008}.
  In particular when $\Gl$ is a quiver algebra arising from a tilting object i.e. $\Gl=\opn{End}_{D^b(X)}(T)$, then we get a deformation of $D^b(X)$ as an $A_{\infty}$-category.
\end{remark}

\subsection{Deformation Quantisation}\label{sec_DQ}
Deforming the multiplication $m\to m+m'$ while maintaining the associativity amounts to demanding $[m+m',m+m']=0$. Remembering that $[m,\sbullet]=-\gd\sbullet$, we get
\bea -\gd m'+\frac12[m',m']=0,\label{MC}\eea
which is called the Maurer-Cartan (MC) equation.
Solving it is not easy due to its non-linearity, but at least one can try a power series solution. Write the deformation as $m\to m+\hbar m_1+\hbar^2 m_2+\cdots$, and unfold the MC-equation order by order in $\hbar$
\bea &&\gd m_1=0,\nn\\
&&\gd m_2=\frac12[m_1,m_1],\nn\\
&&\cdots\label{order_by_order}\eea
We see that already at step two, we may potentially not be able to find $m_2$: though $[m_1,m_1]$ is $\gd$-closed, it is not exact in any obvious way.

The solution to the above problem was provided by Kontsevich \cite{Kontsevich:1997vb}. For simplicity of discussion, we assume $X$ to be a smooth affine variety.
The central role here is played by the Hochschild cohomology of $\Gl$ defined as
\bea H\hspace{-.3em}H^n(\Gl):=\ext^n_{\Gl^e}(\Gl,\Gl).\label{Hochschild_coh}\eea
If one uses the bar resolution, then $\HH^n$ is the $n^{th}$ cohomology of $\B^{\sbullet}$.
An element in $\HH^2$ (more correctly its representative in $\B^2$) gives the $m_1$ in \eqref{order_by_order}, a first order deformation of the multiplication $m$.
The HKR theorem shows that $\HH^2(X)$ is isomorphic (for affine $X$) to $H^0(X,\wedge^2TX)$.
Therefore a bi-vector field $\pi$ leads to a deformation $m_1^{\pi}$
\bea m_1^{\pi}(f,g)=\{f,g\}_{\pi}.\nn\eea
What is special about this particular $m^{\pi}_1$ such that $[m_1^{\pi},m_1^{\pi}]$ is exact?
It was proved that there is an $L_{\infty}$-morphism (section 4.6.2 in \cite{Kontsevich:1997vb})
\bea \wedge^{\sbullet}TX\rightsquigarrow \B^{\sbullet}(\Gl).\label{L_infty_map}\eea
Note that both sides above are dg Lie algebras, however the map in between is an $L_{\infty}$-map, and so has non-linear corrections.
More crucially, under this map, the Maurer-Cartan elements are preserved.
One may now start from a holomorphic Poisson structure $\pi$ on $X$, the fact that $\bar\partial\pi=0=[\pi,\pi]$ means that $\pi$ is an MC-element. Under the $L_{\infty}$-map, it turns into an MC-element $m'(\pi)\in \B^2$. Thus instead of solving \eqref{order_by_order} order by order, one needs to obtain the $L_{\infty}$-map \eqref{L_infty_map}.
This done, the resulting $m+m'(\pi)$ gives a new associative multiplication. The leading deformation in $m'$ corresponds exactly to $\pi$ and so the Poisson bracket.

The detailed formula for the $L_{\infty}$-map is hard to compute, the difficulty again being the large size of $\B^{\sbullet}$. To get around this problem, we turn to an alternative smaller bi-module resolution $\P_{\sbullet}\to\Gl$. We will make use of two such resolutions in sec.\ref{sec_Smdatmr} and \ref{sec_TCSR}, but before that, we will sketch how $\P_{\sbullet}$ can help us.

By virtue of $\P_{\sbullet}$ and $\B_{\sbullet}$ both being a projective resolution for $\Gl$, there must be bi-module maps
\bea  \begin{tikzpicture}
  \matrix (m) [matrix of math nodes, row sep=.5em, column sep=4em]
    { \P_{\sbullet} & \B_{\sbullet} \\        };
  \path[->, font=\scriptsize, bend left =5]
  (m-1-1) edge node[above] {$F_{\sbullet}$} (m-1-2)
  (m-1-2) edge node[below] {$G_{\sbullet}$} (m-1-1);
  \end{tikzpicture}\label{map_F_G}\eea
We will see that the maps satisfy
\bea G_nF_n=1,~~~F_nG_n=1-(h_{n-1}\partial+\partial h_n),\nn\eea
where $h_n:~\B_n\to \B_{n+1}$ is a bi-module chain homotopy.
We will construct the homotopy $h_n$ later and we denote with $h^{\sbullet}$ the map of cochain complexes $\B^{\sbullet}\to \B^{\sbullet-1}$.
We now cut the difficult $L_{\infty}$-map $\wedge^{\sbullet}TX\rightsquigarrow \B^{\sbullet}$ into two parts by way of $\P^{\sbullet}$
\bea  
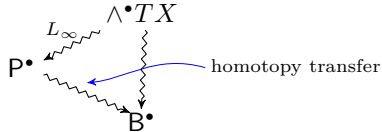
\begin{tikzpicture}
  \matrix (m) [matrix of math nodes, row sep=.7em, column sep=2em]
    {  & \wedge^{\sbullet}TX \\
    \P^{\sbullet} & \\
    & \B^{\sbullet} \\ };
    \draw [->,line join=round,decorate, decoration={zigzag,segment length=4,amplitude=.9,post=lineto,post length=2pt}] (0.7,.5) -- (0.65,-.55);
    \draw [->,line join=round,decorate, decoration={zigzag,segment length=4,amplitude=.9,post=lineto,post length=2pt}] (0.1,.5) -- (-0.65,0.1);
    \draw [->,line join=round,decorate, decoration={zigzag,segment length=4,amplitude=.9,post=lineto,post length=2pt}] (-0.65,-0.1) -- (0.5,-0.6);
    \draw [blue,line width=0.1mm,-{[flex,sep]>}]
        (1.5,0) .. controls (0.8,0.2) and (0.4,-0.2)  .. (-0.1,-0.2);
    \node at (2.7,0) {\scriptsize{homotopy transfer}};
    \node at (-0.4,.5) {\scriptsize{$L_{\infty}$}};
\end{tikzpicture}\nn\eea
In this way, the map between $P^{\sbullet}$ and $B^{\sbullet}$ is a more familiar homotopy transfer, whose formula is given by summing over certain trees. The map between $\wedge^{\sbullet}TX$ and $P^{\sbullet}$ is not straightforward. But we may now choose to solve the MC equation in $P^{\sbullet}$, which is much smaller, so much so that we have some effective combinatorial tools to arrive at the MC-element quickly.

\section{The Minimal Resolution}\label{sec_Smdatmr}
We saw above that computing $\HH^{\sbullet}(\Gl)$ requires resolving $\Gl$ as a bi-module over itself. The bar complex is too big for this and in this section we shall review the minimal resolutions in \cite{BUTLER1999323}.

First, we explain the term `minimal'. Let $(A,\FR{m})$ be a local ring, then a resolution $P_{\sbullet}\to M$ is minimal iff $\partial P_n\subset\FR{m}P_{n-1}$. In the case of a graded algebra $A=\oplus_{n\geq0}A_n$, we require that $\partial P_n\subset A_{>0}P_{n-1}$.
In such cases, if the global dimension of $A$ is finite, then the length of the resolution is the global dimension (see sec 19.1 of \cite{Eisenbud95}).

The graded case is particularly relevant for us since all our quiver algebras are homogeneous, i.e. the ideal of relations $R$ is a homogeneous ideal.
We reiterate the setting: $\Gc=TQ_1$ is the tensor algebra of the length 1 paths $Q_1$ over $S=k\bra e_i\ket$, $J$ is the ideal generated by paths of nonzero length, $R\subset J$ is a homogeneous ideal and $\Gl=\Gc/R$. Thm.7.2. of \cite{BUTLER1999323} gives the following resolution
\bea\cdots\to \Gl\otimes_S\SF{T}_3\otimes_S\Gl\to \Gl\otimes_S\SF{T}_2\otimes_S\Gl\to \Gl\otimes_S\SF{T}_1\otimes_S\Gl\to \Gl\otimes_S\Gl\to \Gl\to0.\label{min_resolution_app}\eea
The various terms in \eqref{min_resolution_app} are
\bea\SF{T}_1=Q_1,~~\SF{T}_2=\frac{R\cap J^2}{JR+RJ},~~\SF{T}_3=\frac{JR\cap RJ}{R^2+JRJ},~~\SF{T}_4=\frac{R^2\cap JRJ}{JR^2+R^2J},\cdots\nn\eea
Also strictly speaking $\SF{T}_3$ is the $S$-bi-module complement of $R^2+JRJ$ in $JR\cap RJ$ and $\SF{T}_4$ is the $S$-bi-module complement of $JR^2+R^2J$ in $R^2\cap JRJ$. Such splits are possible since the algebra $S$ is separable, so we assume a split has been chosen.

\begin{example}\label{Ex_A1}
Consider our running example of the $A_1$ singularity
\bea  \begin{tikzpicture}
  \matrix (m) [matrix of math nodes, row sep=.5em, column sep=5em]
    { e_1 & e_0 \\        };
  \path[->, font=\scriptsize, bend right = 5]
  (m-1-1) edge (m-1-2);
  \path[->, font=\scriptsize, bend left = 5]
  (m-1-1) edge (m-1-2);
  \path[->, blue, font=\scriptsize, bend right = 25]
  (m-1-2) edge (m-1-1);
  \path[->, blue,font=\scriptsize, bend left = 25]
  (m-1-2) edge (m-1-1);

  \node at (0.1,0.2)  {\scriptsize{$w_2$}};
  \node at (0.1,-0.18)  {\scriptsize{$w_1$}};
  \node at (-0.8,0.50)  {\scriptsize{$z_2$}};
  \node at (-0.8,-0.47)  {\scriptsize{$z_1$}};
  \end{tikzpicture},\label{quiver_A1}\eea
with relations
\bea R=\bra e_0w_iz^ie_0,e_1z^iw_ie_1\ket.\nn\eea
Note that our arrows go from right to left e.g. $e_1z^1w_2e_1$ denotes $e_1\stackrel{z^1}{\ot}e_0\stackrel{w_2}{\ot}e_1$.
This is an algebra of global dimension 2, no need to compute $\SF{T}_3$, the resolution reads
\bea &&0\to \Gl e_0\otimes w_iz^i\otimes e_0\Gl{\vcenter{\hbox{\small$\bigoplus$}}}\Gl e_1\otimes z^iw_i\otimes e_1\Gl
\to\nn\\
&&\hspace{2cm}\to \Gl e_1\otimes z^i\otimes e_0\Gl{\vcenter{\hbox{\small$\bigoplus$}}}\Gl e_0\otimes w_i\otimes e_1\Gl\to
\Gl e_0\otimes e_0\Gl{\vcenter{\hbox{\small$\bigoplus$}}}\Gl e_1\otimes e_1\Gl\to \Gl\to0\label{min_res_A1}\eea
the differentials are as follows
\bea
&&\partial(e_0\otimes w_iz^i\otimes e_0)=e_0w_i\otimes z^i\otimes e_0+e_0\otimes w_i\otimes z^ie_0,\nn\\
&&\partial(e_1\otimes z^iw_i\otimes e_1)=e_1z^i\otimes w_i\otimes e_1+e_1\otimes z^i\otimes w_ie_1,\nn\\
&&\partial(e_1\otimes z^i\otimes e_0)=e_1z^i\otimes e_0-e_1\otimes z^ie_0,\nn\\
&&\partial(e_0\otimes w_i\otimes e_1)=e_0w_i\otimes e_1-e_0\otimes w_ie_1,\nn\\
&&\partial(e_a\otimes e_a)=e_a,~~a=0,1.\label{min_res_A1_maps}\eea
\end{example}
In fact, the first two steps of the bi-module resolution of a quiver algebra look identical: the $\SF{T}_1$ term is spanned by paths of length 1, with the differential
\bea \partial(1\otimes p\otimes 1)=p\otimes 1-1\otimes p.\nn\eea
For the $\SF{T}_2$ term, we pick a generating set for the relations $R$. For one such generator $p_1\cdots p_n\in R$, the differential acts as the split
\bea \partial(1\otimes p_1\cdots p_n\otimes 1)=p_1\cdots p_{n-1}\otimes p_n\otimes 1+p_1\cdots p_{n-2}\otimes p_{n-1}\otimes p_n+\cdots+1\otimes p_1\otimes p_2\cdots p_n.\nn\eea
Equipped with this bit of information, we give also the example of $A_n$ and $D_4$ Kleinian singularity
\begin{example}
 The $A_n$ singularity is defined as the locus $f=xy-z^{n+1}=0$. According to \cite{GNTjurina_1969}, it has an $n$-dim family of local semi-universal flat deformations, with basis
 \bea \frac{\BB{C}[x,y,z]}{\bra f,\partial_xf,\partial_yf,\partial_zf\ket}.\nn\eea
 Thus the deformations are written as
 \bea f\to f_t=xy-z^{n+1}-t_0-t_1z-\cdots t_{n-1}z^{n-1}.\nn\eea
 For generic $t$, the surface $X_t$ cut out by $f_t$ is smooth, so is its real slice $M_t$. The tilting theory applied to $\{f=0\}$ will give a non-commutative resolution to the same singularity. Running deformation quantisation afterwards will be regarded as applying brane quantisation to the complexification $M_t\hookrightarrow X_t$.

 We say a quick word about how the tilting object is obtained, more details can be found in the excellent lecture notes by Wemyss \cite{NCCR4}.
 Set $A$ to be the algebra $\BB{C}[x,y,z]/(xy-z^{n+1})$ and consider its ideals $I_p=(x,z^p)$ with $p=1,\cdots,n$. Between $I_p$ and $A$ we have the maps
\bea
\begin{tikzpicture}
  \matrix (m) [matrix of math nodes, row sep=.5em, column sep=3em]
    { A & (x,z) \\        };
  \path[->, font=\scriptsize, bend left =5]
  (m-1-1) edge node[above] {$z$} (m-1-2);
  \path[>->, font=\scriptsize, bend left =5]
  (m-1-2) edge (m-1-1);
  \end{tikzpicture}~~~
 \begin{tikzpicture}
  \matrix (m) [matrix of math nodes, row sep=.5em, column sep=3em]
    { (x,z^p) & (x,z^{p+1}) \\        };
  \path[->, font=\scriptsize, bend left =5]
  (m-1-1) edge node[above] {$z$} (m-1-2);
  \path[>->, font=\scriptsize, bend left =5]
  (m-1-2) edge (m-1-1);
  \end{tikzpicture}~~~\begin{tikzpicture}
  \matrix (m) [matrix of math nodes, row sep=.5em, column sep=3em]
    { (x,z^n) & A \\        };
  \path[->, font=\scriptsize, bend left =5]
  (m-1-1) edge node[above] {$y/z^n$} (m-1-2)
  (m-1-2) edge node[below] {$x$} (m-1-1);
  \end{tikzpicture}\label{map_A_n}\eea
  We put these modules and maps into the following quiver
  \bea
   \begin{tikzpicture}
   \matrix (m) [matrix of math nodes, row sep=2.5em, column sep=3em]
     {    & e_0 & \\
         e_1 & \cdots & e_n \\ };
\path[->,bend left=10]
   (m-2-1) edge  (m-1-2)
   (m-1-2) edge  (m-2-1)
   (m-1-2) edge  (m-2-3)
   (m-2-3) edge  (m-1-2)
   (m-2-3) edge  (m-2-2)
   (m-2-2) edge  (m-2-3)
   (m-2-2) edge  (m-2-1)
   (m-2-1) edge  (m-2-2);
\end{tikzpicture}\label{quiver_An}\eea
with $e_0$ corresponding to $A$ and $e_p$ to $I_p$. From \eqref{map_A_n}, we have the following relations among the maps
\bea (e_i\ot e_{i+1}\ot e_i)=(e_i\ot e_{i-1}\ot e_i).\label{relation_A_n}\eea
The reader may write down the bi-module resolution himself.
\end{example}
\begin{example}\label{Ex_D4}
  For a bit of change, we also take the $D_4$ singularity: $X=\{(u,v,w)\in\BB{C}^3|f=v^2-u^2w+4w^3=0\}$.
  To get the quiver requires some work. Let us recall that the singularity defined by $f=0$ also arises out of taking quotient of $\BB{C}^2$ by a finite subgroup of $SU(2)$, known as $BD_4$ (binary $D_4$). It has generators
  \bea T=\left[
         \begin{array}{cc}
           i & 0 \\
           0 & -i \\
         \end{array}\right],~~~S=\left[
         \begin{array}{cc}
           0 & i \\
           i & 0 \\
         \end{array}\right].\nn\eea
  This group has four non-trivial irreducible representations, one of which is denoted $R$, the fundamental representation of $SU(2)$.
  If we denote the two basis vectors of $R$ as $|+\ket$ and $|-\ket$, then the remaining three 1D representations have basis
  \bea R_1:~|+-\ket+|-+\ket,~~~R_2:~|++\ket+|--\ket,~~~R_3:~|++\ket-|--\ket,\nn\eea
  while the coordinates $(x,y)$ of $\BB{C}^2$ sit in the representation $R^*$. A natural object to study is the set of polynomials invariant under $G=BD_4$. These are generated by
  \bea u=x^4+y^4,~~~v=x^5y -y^5x,~~~w=x^2y^2\nn\eea
  which satisfy precisely the relation $f=v^2-u^2w+4w^3=0$. One can also consider the relative invariants $(\BB{C}[x,y]\otimes R_i)^G$ for each irreducible representation. We collect the results here
  \bea &&(\BB{C}[x,y]\otimes R_1)^G=\bra xy,x^4-y^4\ket,\nn\\
  &&(\BB{C}[x,y]\otimes R_2)^G=\bra x^2+y^2,x^3y-y^3x\ket,~~~(\BB{C}[x,y]\otimes R_3)^G=\bra x^2-y^2,x^3y+y^3x\ket,\nn\\
  &&(\BB{C}[x,y]\otimes R)^G=\bra
  \begin{bmatrix}
    x \\ y
  \end{bmatrix},\begin{bmatrix}
    x^5 \\ y^5
  \end{bmatrix},\begin{bmatrix}
    x^2y \\ -xy^2
  \end{bmatrix},\begin{bmatrix}
    y^3 \\ -x^3
  \end{bmatrix}\ket.\nn\eea
  Setting $R_0$ to be the trivial representation, we arrange these modules into the affine Dynkin diagram $D_4$
  \bea
   \begin{tikzpicture}
   \matrix (m) [matrix of math nodes, row sep=3em, column sep=4em]
     {    & R_0 & \\
          R_1 & R & R_3 \\
            & R_2 &  \\ };
\path[->,bend left=10]
   (m-2-1) edge node[above] {$[-x;y]$} (m-2-2)
   (m-2-3) edge node[below] {$[y;x]$} (m-2-2)
   (m-1-2) edge node[right] {$[x;y]$} (m-2-2)
   (m-3-2) edge node[left] {$[y;-x]$} (m-2-2);
\path[->,bend left=10]
   (m-2-2) edge (m-2-1)
   (m-2-2) edge (m-2-3)
   (m-2-2) edge (m-1-2)
   (m-2-2) edge (m-3-2);
\end{tikzpicture}\label{D4_quiver}\eea
where the maps from $R_{0,1,2,3}$ to $R$ have been labelled. The maps $R\to R_{1,2,3,0}$ are obtained by tensoring $R$ with $[x;y]$ and decomposing into the various relative invariants. We summarise the results
\bea
\begin{bmatrix}
  f(x,y) \\ g(x,y)
\end{bmatrix}&\to&\frac{1}{2}(xg+yf)(|+-\ket+|-+\ket)+\frac12(xf+yg)(|++\ket+|--\ket)\nn\\
&&+\frac12(xf-yg)(|++\ket-|--\ket)+\frac12(xg-yf)(|+-\ket-|-+\ket)\label{D4_maps}\eea
where the rhs are the invariants relative to $R_1,R_2,R_3,R_0$ respectively.

As for the relations among the maps, one can easily be convinced that $R_i\ot R\ot R_i$ is zero for all $i$, and
\bea \sum_i R\ot R_i\ot R=0.\nn\eea
We remark that these relations are the same ones that appeared in Kronheimer's HyperK\"ahler quotient construction of the ALE spaces \cite{Kronheimer1989TheCO}.
Indeed, if one starts from the quiver \eqref{D4_quiver}, and considers the quiver representation with dimension vector $\udl{d}=(1,1,1,1,2)$ for the nodes $R_{0,1,2,3},R$, then the moduli space of such representations can be obtained as a GIT or a HyperK\"ahler quotient, in which the relations between the arrows are the moment map conditions. This HyperK\"ahler quotient gives a surface that is a crepant resolution of the $D_4$ singularity $\hat X\to X$, whereas the quiver algebra \eqref{D4_quiver} gives a non-commutative crepant resolution (NCCR).
\end{example}

\vskip.3cm

Finally we will recall the central example $X=T^*\BB{P}^2$. The quiver algebra is obtained as the endomorphism algebra of the tilting object of $D^b(X)$, which consists of ${\cal O}_{\BB{P}^2}$, ${\cal O}_{\BB{P}^2}(1)$ and ${\cal O}_{\BB{P}^2}(2)$ (then pulled back to $T^*\BB{P}^2$), see \cite{TODA20101}.
The quiver reads
\bea  \begin{tikzpicture}
  \matrix (m) [matrix of math nodes, row sep=.5em, column sep=5em]
    { 2 & 1 & 0 \\        };
  \path[->, font=\scriptsize, bend right = 10] (m-1-1) edge node [below] {\scriptsize{$\vec w\times 3$}} (m-1-2);
  \path[->, font=\scriptsize, bend right = 10] (m-1-2) edge node [below] {\scriptsize{$\vec w\times 3$}} (m-1-3);
  \path[->, blue,font=\scriptsize, bend right = 10] (m-1-2) edge node [above] {\scriptsize{$\vec z\times 3$}} (m-1-1);
  \path[->, blue,font=\scriptsize, bend right = 10] (m-1-3) edge node [above] {\scriptsize{$\vec z\times 3$}} (m-1-2);
\end{tikzpicture}\label{quiver_TP2}\eea
where each arrow is in fact three arrows labelled with $z^{1,2,3}$ or $w_{1,2,3}$.
Here the three nodes correspond to the object ${\cal O}_{\BB{P}^2}(0,1,2)$ respectively. Note that for the tilting theory to work, it is crucial that the higher cohomologies $\ext^{>0}({\cal O}(m),{\cal O}(n))$ vanish for $m,n=0,1,2$.

The relations of the quiver are
\bea &e_2\stackrel{z^i}{\leftarrow}\stackrel{z^j}{\leftarrow}e_0=e_2\stackrel{z^j}{\leftarrow}\stackrel{z^i}{\leftarrow}e_0,~~~
e_0\stackrel{w_i}{\leftarrow}\stackrel{w_j}{\leftarrow}e_2=e_0\stackrel{w_j}{\leftarrow}\stackrel{w_i}{\leftarrow}e_2,~~~
e_1\stackrel{z^i}{\leftarrow}\stackrel{w_j}{\leftarrow}e_1=e_1\stackrel{w_j}{\leftarrow}\stackrel{z^i}{\leftarrow}e_1,\nn\\
&\sum_ie_a\stackrel{z^i}{\leftarrow}\stackrel{w_i}{\leftarrow}e_a=0,~~~a=1,2,~~~~
\sum_ie_a\stackrel{w_i}{\leftarrow}\stackrel{z^i}{\leftarrow}e_a=0,~~~a=0,1.\nn\eea
The algebra is of global dimension 4, so one does have to work out the $\SF{T}_{3,4}$ terms in \eqref{min_resolution_app}. This was done in  sec.6.3 and appendix B in \cite{RW_loc_tilt}
\bea &&\SF{T}_2=\left[
           \begin{array}{c}
             e_1(z\cdotp w\oplus w\cdotp z)e_1 \\
             e_2(z\cdotp w)e_2 \\
             e_0(w\cdotp z)e_0 \\
             e_1([w_i,z^j]_{\rm t.f.})e_1 \\
             e_2(z^{[i}z^{j]})e_0 \\
             e_0(w_{[i}w_{j]})e_2 \\
           \end{array}\right]~~~\SF{T}_3=\left[
           \begin{array}{c}
             e_2(z^kw_kz^j-z^kz^jw_k+z^jz^kw_k)e_1 \\
             e_0(w_kz^kw_j-w_kw_jz^k+w_jw_kz^k)e_1 \\
             e_1(z^jw_kz^k-w_kz^jz^k+w_kz^kz^j)e_0 \\
             e_1(w_jz^kw_k-z^kw_jw_k+z^kw_kw_j)e_2
           \end{array}\right],~\nn\\
           &&
           ~~\SF{T}_4=\left[
           \begin{array}{c}
             e_2(z^{[k}z^{l]}w_{[k}w_{l]}-2z^kw_kz^lw_l)e_2 \\
             e_0(w_{[k}w_{l]}z^{[k}z^{l]}-2w_kz^kw_lz^l)e_0 \\
             e_1(\left[z^i,w_j\right]\left[w_i,z^j\right]-z^iw_iw_jz^j-w_iz^iz^jw_j)e_1
\end{array}\right].\label{T_1234}\eea
As for the differential, we have mentioned that $\partial \SF{T}_1$ and $\partial \SF{T}_2$ are standard for all cases.
As for $\partial \SF{T}_3$, by definition, a length 3 path in $\SF{T}_3$ can be written in two ways $p=xr=r'x'$ where $x,x'$ are two length 1 paths and $r,r'$ are two length 2 relations. Then $\partial_3(1\otimes_Sp\otimes_S1)$ reads
\bea \Gl\otimes_S\SF{T}_3\otimes_S\Gl \ni 1\otimes_Sp\otimes_S1 \stackrel{\partial_3}{\to} x\otimes_Sr\otimes_S1-1\otimes_Sr'\otimes_Sx'\in \Gl\otimes_S\SF{T}_2\otimes_S\Gl,\nn\eea
Finally a path $q$ of length 4 in $\SF{T}_4$ is written in two ways $q=rr'=xr''y$ with $x,y$ of length 1 and $r,r',r''$ being relations. We have the map
\bea \Gl\otimes_S\SF{T}_4\otimes_S\Gl \ni 1\otimes_Sq\otimes_S1 \stackrel{\partial_4}{\to} x\otimes_Sr''y\otimes_S1+1\otimes_Sxr''\otimes_S y\in \Gl\otimes_S\SF{T}_3\otimes_S\Gl.\nn\eea
Note that prop.A.1 (b) of \cite{BUTLER1999323} guarantee that $r''y$ and $xr''$ do belong to $\SF{T}_3$.
Let us look at an example $q=e_2(z^{[k}z^{l]}w_{[k}w_{l]}-2z^kw_kz^lw_l)e_2$, then
\bea \partial_4(1\otimes_Sq\otimes_S1)=2z^k\otimes_S([z^l,w_k]w_l-z^lw_lw_k)\otimes_S1+2\otimes_S(z^k[z^l,w_k]-z^lz^kw_k)\otimes_Sw_l\nn\eea
and each parenthesis is manifestly in $JR\cap RJ$ and so the rhs is in $\SF{T}_3$.

\section{The Chouhy-Solotar Resolution}\label{sec_TCSR}
In contrast to the resolution of the previous section, the one constructed in \cite{CHOUHY201522} is not always minimal. However it is very down-to-the-earth and closer to human thinking. Let us spend a few words on this very aspect. Imagine one proposes an algebra presented as generators and relations, i.e. the algebra elements are words in certain alphabet. One can apply the relations to reduce any word until one runs out of relations. This process is by no means unique: depending on how one does the reduction, the resulting irreducible word may be different. One can also define the product of two irreducible words by juxtaposing them and then reducing them, but this will likewise not give a unique result.

The failure of the uniqueness in the reduction comes from two sources.
First, one should well order the set of words (e.g.  dictionary order), and try to reduce the order of a word at each step of the reduction. But one may fall into `local minima' (we will see some examples soon). Second, there may be words like `abc', such that both `ab' and `bc' can be reduced, and so the result depends on the order of reduction. Such words are known as ambiguities. In fact, these ambiguities are very similar to the $\SF{T}_3$ term in \eqref{min_resolution_app}.
There might be higher ambiguities, i.e. ambiguities of ambiguities e.g. the $\SF{T}_4$ terms in \eqref{min_resolution_app}, and so on.
One can formalise these ideas and build a resolution
\bea \SF{P}_{\sbullet}:~~ \cdots\to \Gl\otimes\SF{A}_2\otimes\Gl \to \Gl\otimes\SF{A}_1\otimes\Gl\to\Gl\otimes\SF{A}_0\otimes\Gl\to \Gl\otimes\Gl\to \Gl\to0,\label{non_min_res}\eea
the $\SF{A}_0$ term is isomorphic to $\SF{T}_1$, but $\SF{A}_1$ may already need more generators than $\SF{T}_2$, for a reason that we shall review next. We have also named this resolution $\SF{P}_{\sbullet}$. This is because all the Chouhy-Solotar resolution for us will have an obvious contraction to the minimal resolution and the identical notation seems justified.

\subsection{Reduction System}\label{sec_RS}
We first give a well-ordering $<_{\go}$ to the set of words, the reader should consult page 5 of \cite{CHOUHY201522} for the most general definition. But for this paper, we can use a simpler version. First $w_1<_{\go} w_2$ if $w_1$ has shorter length, or, when the lengths are the same, $w_1$ precedes $w_2$ in the dictionary. Let us reiterate that our words are read \emph{from right to left}.
This ordering is sufficient if the relations preserve or decrease the word length.
\begin{definition}(Def 2.8 \cite{CHOUHY201522})
  Let $p\in kQ$ be a path, written as a sum of monomials. The tip, denoted $\SF{tip}(p)$, is the largest monomial according to the order $<_{\go}$.

  Given an ideal of $kQ$, the set $S(I)$ (or just $S$ if $I$ is clear) consists of monomials $p$ that are tips of elements of $I$, with the condition that it does not contain a proper sub-path which is a tip of $I$.
\end{definition}
From the definition, a path in $S(I)$ cannot be the sub-path of another element in $S(I)$. Further, by the definition, given $s\in S$, there is another path $f_s$ such that $s-f_s\in I$, and $f_s<_{\go}s$.
\begin{remark}
  The set $S$ is somewhat like a Gr\"obner basis and in general not easy to compute. It is not enough to just choose a set of generators for $I$, and take the tip of each generator. This typically results in a set smaller than $S$.

  One can use $\{(s,f_s)|s\in S\}$ as a \emph{reduction system}: one scans a word for sub-paths coinciding with $s$, then replace $s$ with $f_s$.
  Lem 2.9 in \cite{CHOUHY201522} shows that the set $s-f_s$ generates $I$, and lem 2.10 shows that the reduction of a word is unique.
\end{remark}
We give an example to show why the naive way of computing $S$ described above does not work.
\begin{example}\label{Ex_A1_I}
  Back to our $A_1$ example \ref{Ex_A1}, we order the letters
  \bea w_1>w_2>z^2>z^1\label{order_best}\eea
  where we omit the $\go$ in $>_{\go}$. The ideal of relations $R=\bra e_1z^iw_ie_1,e_0w_iz^ie_0\ket$ has the set of tips $S=\{e_1z^1w_1e_1,e_0w_2z^2e_0\}$. The reduction of the word $e_1z^1w_2z^2e_0$ reads
  \bea e_1z^1\underline{w_2z^2}e_0\to -e_1\underline{z^1w_1}z^1e_0\to e_1z^2w_2z^1e_0.\nn\eea
  In fact the reduction of a more general word results in ordering the $z$'s according to the order above, and likewise for the $w$'s.

  For illustration purpose, we deliberately pick a \emph{sub-optimal} order
  \bea w_2>w_1>z^2>z^1.\label{order_bad}\eea
  If pick $S$ in the naive way as $S=\{e_1z^2w_2e_1,e_0w_2z^2e_0\}$, then the reduction of $e_1z^2w_2z^2e_0$ will depend on how one does it
  \bea e_1z^2\underline{w_2z^2}e_0\to -e_1z^2w_1z^1e_0,\nn\\
  e_1\underline{z^2w_2}z^2e_0\to -e_1z^1w_1z^2e_0.\nn\eea
  The two words in the rhs are no longer reducible, but are not equal as they ought to be. What went wrong here is: the two words are still equal up to $R$:
  \bea e_1(z^2w_1z^1-z^1w_1z^2)e_0=\sum_ie_1(z^2w_iz^i-z^iw_iz^2)e_0.\in R\nn\eea
  The tip of the first term is $z^2w_2z^2$, so is it for the second term. Thus the `leading' tips get cancelled, leaving the sub-leading tip $z^1w_1z^2$. This is the `local minima' that we mentioned earlier. One must accordingly include this sub-leading tip into $S$
  \bea S=\bra e_1z^2w_2e_1,e_0w_2z^2e_0,e_1z^1w_1z^2e_0,e_0w_1z^1w_2e_1\ket\nn\eea
  giving a reduction system
  \bea z^2w_2\to -z^1w_1,~~~w_2z^2\to -w_1z^1,~~~z^1w_1z^2\to z^2w_1z^1,~~~w_1z^1w_2\to w_2z^1w_1.\nn\eea
  We have also stopped writing the $e_i$'s wherever they can be inferred from the word itself.
\end{example}

Once we obtain $S$, the term $\SF{A}_1$ in the resolution \eqref{non_min_res} is precisely generated by the set $S$.
To continue, we introduce formally the $n$-ambiguities
\begin{definition}\label{def_ambiguity}(Def 3.1 \cite{CHOUHY201522})
  A word $u_0u_1\cdots u_n$ is a left $n$-ambiguity ($n\geq2$) if: $u_i,~i=1,\cdots,n$ are irreducible and $u_0\in Q_1$ is such that $u_iu_{i+1}$ is reducible but $u_id$ is not, with $d$ being any proper sub-word starting from the left of $u_{i+1}$.

  There is a symmetric notion of right $n$-ambiguity, but it turns out that the set of left and right ambiguities coincide.
\end{definition}

\begin{example}\label{Ex_A1_II}
  Continuing from Ex.\ref{Ex_A1_I}, the following are the left 2,3,4-ambiguities
  \bea &&\SF{A}_2:\;\bra w_2\cdotp z^2\cdotp w_2,~~z^2\cdotp w_2\cdotp z^2\ket,\nn\\
  &&\SF{A}_3:\;\bra z^1\cdotp w_1z^2\cdotp w_2\cdotp z^2,w_1\cdotp z^1w_2\cdotp z^2\cdotp w_2,z^2\cdotp w_2\cdotp z^2\cdotp w_2,w_2\cdotp z^2\cdotp w_2\cdotp z^2\ket,\nn\\
  &&\cdots\nn\eea
  where we have inserted $\cdotp$ to show how the words are broken into ambiguities. For the higher $\SF{A}_n$'s, the pattern is that one adds $z^2$ or $w_2$ to the right of $\SF{A}_{n-1}$ where appropriate.
\end{example}

Now we build the resolution from $\SF{A}_n$ as in \eqref{non_min_res}
\bea \cdots\to \Gl\otimes\SF{A}_2\otimes\Gl \stackrel{\partial_3}{\to} \Gl\otimes\SF{A}_1\otimes\Gl\stackrel{\partial_2}{\to}\Gl\otimes\SF{A}_0\otimes\Gl\stackrel{\partial_1}{\to} \Gl\otimes\Gl\stackrel{\ep}{\to} \Gl\to0,\label{non_min_res_map}\eea
and we define the chain maps next. The explicit maps in low degrees are given in sec 6 in \cite{CHOUHY201522}. Here we shall describe these maps to spare the reader the need to plough through the notations.

First the augmentation map $\ep:\;\Gl\otimes\Gl\to \Gl$ is just the multiplication. The map $\partial_1$ is identical to the minimal resolution $\partial_1(1\otimes a\otimes 1)=a\otimes 1-1\otimes a$, where $a\in Q_1$.

To write $\partial_2$, we define a few operations on the words. First $\opn{spl}$ is the split
\bea \opn{spl}(c_n\cdots c_1)=\sum_{i=1}^n\overline{c_n\cdots c_{i+1}}\otimes c_i\otimes \overline{c_{i-1}\cdots c_1},\label{split_map}\eea
We caution the reader that the factor to the left and right of the tensor are in $\Gl$, and so a reduction mod $R$ is applied to $c_n\cdots c_{i+1}$ and $c_{i-1}\cdots c_1$ automatically. We also place a bar on top to stress this point.
Take $1\otimes s\otimes 1\in \SF{A}_1$, set
\bea \partial_2(a\otimes s\otimes b)=a\opn{spl}(s-f_s)b.\nn\eea
Since $s-f_s\in R$, it is easily seen that $\partial_1\partial_2=0$.

To define $\partial_3$, we make use of the right reduction map $\opn{red}^R$. This map does exactly what it says on the tin: reducing a word from the right, i.e. one scans for the rightmost segment of a word that matches a tip, say $s\in S$, then replaces $s$ with the corresponding $f_s$, and starts over again until no more tips are found.
Somewhat awkwardly, we will need another reduction map $\opn{red}^{L'}$, which does the first reduction from the left and the remaining ones from the right.

Now one applies $\opn{red}^{L'}$ to $\SF{A}_2$ (ambiguities). Along the way, one collects all the tips used during the reduction. For example, when reducing $asb$ to $af_sb$, one collects a term $a\otimes s\otimes b\in \SF{A}_1$. Do the same using $\opn{red}^R$ and one gets another set of tips. Then we get $\partial_3$ as
\bea \partial_3(a\otimes w\otimes b)=a\big(\textrm{tips from }\opn{red}^{L'}-\textrm{tips from }\opn{red}^R\big)b.\label{partial_3_P}\eea
\begin{remark}
One can readily appreciate the intuitive idea behind $\partial_3$. Given an ambiguity $w$, we know $\opn{red}^{L'}(w)=\opn{red}^R(w)$, but one will collect a different set of tips in doing reductions in two different orders. Therefore $\partial_3(\SF{A}_2)$ gives the relation between relations.

The map $\opn{red}^{L'}$ seems somewhat awkward. But in fact, it is only important that the first reduction in $\opn{red}^{L'}$ be from the left. One can even choose a different set of reductions for each element of $\SF{A}_2$, see prop.6.6 in \cite{CHOUHY201522}.
Here we follow the choice of \cite{BarmeierWang}, its convenience will be clear later.
\end{remark}
\begin{example}\label{Ex_A1_III}
  Continuing from Ex.\ref{Ex_A1_II}, the differentials acting on relations are
  \bea &&\SF{A}_1\ni z^2w_2\stackrel{\partial_2}{\to}z^i\otimes w_i\otimes e_1+e_1\otimes z^i\otimes w_i,\nn\\
  &&\SF{A}_1\ni z^1w_1z^2\stackrel{\partial_2}{\to}z^{[1}w_1\otimes z^{2]}\otimes e_0+z^{[1}\otimes w_1\otimes z^{2]}+e_1\otimes z^{[1}\otimes w_1z^{2]}.\nn\eea
  For $\partial$ acting on ambiguities, we shall give a sample calculation
  \bea z^2w_2z^2&\stackrel{\opn{red}^{L'}}{\to}& -z^1w_1z^2 \to -z^2w_1z^1,\nn\\
  &\stackrel{\opn{red}^R}{\to}& -z^2w_1z^1.\nn\eea
  Collecting the tips of the relations used along the way, we get
  \bea \partial(z^2w_2z^2)=e_1\otimes z^2w_2\otimes z^1-{\color{blue}e_1\otimes z^1w_1z^2\otimes e_0}-z^2\otimes w_2z^2\otimes e_0,\nn\eea
  We would like to draw the reader's attention to the blue term. It does not lie in $\Gl_{>0}\SF{A}_1+\SF{A}_1\Gl_{>0}$, and so the resolution is not minimal,
  see the beginning of sec.\ref{sec_Smdatmr}.
  What transpired is that, with the sub-optimal ordering, the term $z^1w_1z^2$ is the tip of a superfluous reduction $z^1w_1z^2\mapsto z^2w_1z^1$, needed solely to give the lhs a leg up the `local minima'. On the other hand, we just found a term in $z^2w_2z^2\in\SF{A}_2$ that maps isomorphically onto the said superfluous tip. One can use a retraction to `cancel' the two terms from $\SF{A}_2$ and $\SF{A}_1$ in pairs.

  We take another example, a four letter ambiguity $w_1\cdotp z^1w_2\cdotp z^2$,
  \bea w_1z^1w_2z^2&\stackrel{\opn{red}^{L'}}{\to}&w_2z^1w_1z^2\to w_2z^2w_1z^1\to -w_1z^1w_1z^1,\nn\\
  &\stackrel{\opn{red}^R}{\to}&-w_1z^1w_1z^1.\nn\eea
  Collecting the tips gives
  \bea
  \partial(w_1z^1w_2z^2)=e_0\otimes w_1z^1w_2\otimes z^2+w_2\otimes z^1w_1z^2\otimes e_0+e_0\otimes w_2z^2\otimes w_1z^1-w_1z^1\otimes w_2z^2\otimes e_0.\nn\eea
  This term is in turn cancelled by a 3-ambiguity $w_2\cdotp z^2\cdotp w_2\cdotp z^2$.
  \begin{remark}\label{rmk_contract}
  In fact, this pattern of cancellation of all terms superfluous persists: even as our resolution $\SF{A}_n$ is nonzero for all $n$, there is an easy chain homotopy that contracts it to the minimal resolution \eqref{min_res_A1}.
  This contraction grants us the latitude of treating the current non-minimal resolution as a minimal one.
  \end{remark}
\end{example}

\subsection{Resolution for $T^*\BB{P}^2$}\label{sec_RfT}
We finish this technical section by giving the $\SF{A}_1$ and $\SF{A}_2$ terms for the quiver \eqref{quiver_TP2} relevant for $T^*\BB{P}^2$.

We first pick an order for the letters
\bea z^1<z^2<z^3<w_1<w_2<w_3.\nn\eea
The reduction system consists of the `essential ones' and `superfluous ones', the former being in 1-1 correspondence with $\SF{T}_2$, while the latter serves to resolve the `local minima'. It is possible that by experimenting with slightly different ordering, the amount of superfluous terms may go down, but we have not invested any effort in that regard.

Starting with the `essential part' of the reduction system
\bea &&\bra z^iz^j\mapsto z^jz^i\ket_{20},~~~~\bra w_jw_i\mapsto w_iw_j\ket_{02},~~~i<j,\nn\\
&&\bra z^iw_j\mapsto w_jz^i\ket_{11},~~i\neq j,~~~~\bra z^1w_1\mapsto w_1z^1\ket_{11},~~\bra z^2w_2\mapsto w_2z^2\ket_{11}\nn\\
&&\bra z^3w_3\mapsto -w_1z^1-w_2z^2\ket_{11},\nn\\
&&\bra w_3z^3\mapsto -w_1z^1-w_2z^2\ket_{11},\nn\\
&&\bra z^3w_3\mapsto -z^1w_1-z^2w_2\ket_{22},\nn\\
&&\bra w_3z^3\mapsto -w_1z^1-w_2z^2\ket_{00},\label{rel_nec}\eea
where the notation $\bra\cdots\ket_{ij}$ is a space saving device for $e_i\otimes(\cdots)\otimes e_j$. Every term here has a clear counterpart in $\SF{T}_2$ of \eqref{T_1234}, and hence essential.

Now the `superfluous part'
\bea && \bra w_2z^1w_3\mapsto w_3z^1w_2\ket_{12},~~~\bra w_1z^2w_3\mapsto w_3z^2w_1\ket_{12},~~~\bra w_1z^3w_2\mapsto w_2z^3w_1\ket_{12},\nn\\
&&\bra w_1z^1w_2\mapsto w_2z^1w_1\ket_{12},~~~\bra w_1z^2w_2\mapsto w_2z^2w_1\ket_{12},~~~\bra w_2z^2w_3\mapsto w_3z^2w_2\ket_{12},\nn\\
&&\bra w_1z^1w_3\mapsto w_3z^1w_1\ket_{12}.\label{rel_wzw}\\
&&\bra z^2w_1z^3\mapsto z^3w_1z^2\ket_{21},~~~\bra z^1w_2z^3\mapsto z^3w_2z^1\ket_{21},~~~\bra z^1w_3z^2\mapsto z^2w_3z^1\ket_{21},\nn\\
&&\bra z^2w_2z^3\mapsto z^3w_2z^2\ket_{21},~~~\bra z^1w_1z^3\mapsto z^3w_1z^1\ket_{21},~~~\bra z^1w_1z^2\mapsto z^2w_1z^1\ket_{21},\nn\\
&&\bra z^1w_2z^2\mapsto z^2w_2z^1\ket_{21}\label{rel_zwz}.\\
&&\bra w_3w_1z^3\mapsto -w_2w_1z^2-w_1w_1z^1\ket_{01},\nn\\
&&\bra w_3w_2z^3\mapsto -w_2w_2z^2-w_2w_1z^1\ket_{01}.\label{new_rel}\eea
These will get mapped onto by terms from $\SF{A}_2$, to which we turn now.

The $\SF{A}_2$ terms have likewise an essential and a superfluous part.
Our earlier remark \ref{rmk_contract} about the contraction equally applies.
From this understanding, when listing $\SF{A}_2$, we will need, other than those corresponding to $\SF{T}_3$, only those that will serve to cancel the superfluous terms in $\SF{A}_1$.
We have
\bea &&\bra z^iz^jw_k\ket_{21},~~~~\bra z^3w_3z^3\ket_{21}~~~~i<j,\nn\\
&&\bra z^kw_iw_j\ket_{12},~~~~\bra w_3z^3w_3\ket_{12}~~~~i<j,\nn\\
&&\bra z^iw_3z^3\ket_{10},\nn\\
&&\bra w_3z^3w_i\ket_{01},~~~~\bra w_{1,2}w_3z^3\ket_{01}\label{Amb}.\eea
We summarise the differentials of $\SF{A}_2$ in the next lemma without proof, which is but some mind-numbing calculations.
\begin{lemma}\label{prop_diff_TP2}
The 28 ambiguities in \eqref{Amb} map under $\partial$ to $\partial\SF{T}_3$ (12 in total), and onto the superfluous relations Eqs.\ref{rel_wzw}, \ref{rel_zwz}, \ref{new_rel} (16 in total). Concretely
  \bea &&\partial\bra z^iw_3z^3\ket_{10}\sim-\partial(\SF{T}_3^3)^i,~~~i=1,2,3,\nn\\
&& \partial\bra w_3z^3w_3\ket_{01}\sim-\partial(\SF{T}_3^2)_3,\nn\\
&&\partial\bra w_3z^3w_i+w_iw_3z^3\ket_{01}\sim-\partial(\SF{T}_3^2)_i,~~i=1,2\nn\\
&&\partial\bra z^1z^2w_2+z^1z^3w_3\ket_{21}\sim-\partial(\SF{T}_3^1)^1,\nn\\
&&\partial\bra z^1z^2w_1- z^2z^3w_3\ket_{21}\sim\partial(\SF{T}_3^1)^2,\nn\\
&&\partial\bra -z^3w_3z^3+z^1z^3w_1+z^2z^3w_2\ket_{21}\sim\partial(\SF{T}_3^1)^3.\nn\eea
Here the notation $(\SF{T}_3^2)_j$ means the second entry in $\SF{T}_3$ in \eqref{T_1234}, i.e. $\bra w_kz^kw_j-w_kw_jz^k+w_jw_kz^k\ket_{01}$.
\end{lemma}
As an example, we take $\bra z^1w_3z^3\ket_{10}$
\bea &&\bra z^1w_3z^3\ket_{10}\stackrel{\opn{red}^{L'}}{\to} w_3z^1z^3\to w_3z^3z^1\to -(w_1z^1+w_2z^2)z^1,\nn\\
&&\hspace{1.75cm}\stackrel{\opn{red}^R}{\to}-z^1(w_1z^1+w_2z^2)\to -w_1z^1z^1-w_2z^1z^2\to -w_1z^1z^1-w_2z^2z^1,\nn\eea
collecting the tips above, we find
\bea \partial\bra z^1w_3z^3\ket_{10}=\bra z^1w_3\ket_{11}z^3+w_3\bra z^1z^3\ket_{20}+\bra w_3z^3\ket_{11}z^1-z^1\bra w_3z^3\ket_{00}+\bra z^1w_1\ket_{11} z^1
+\bra z^1w_2\ket_{11} z^2+w_2\bra z^1z^2\ket_{20}.\nn\eea
We remind the reader that $\bra\cdots\ket_{ij}:=e_i\otimes \cdots\otimes e_j$. In the brackets, we have written only the tips, and if one completes the tip into relations, e.g. $z^1z^2\sim z^1z^2-z^2z^1$ and $w_3z^3\sim w_iz^i$, one finds that the rhs matches $-\partial(\SF{T}_3^3)_1$.

We take a more involved example.
We start with $\bra z^iz^jw_k\ket_{21}$
\bea &&\bra z^1z^2w_1\ket_{21}\stackrel{\opn{red}^{L'}}{\to} \to z^2z^1w_1\to z^2w_1z^1,\nn\\
&&\hspace{1.75cm}\stackrel{\opn{red}^R}{\to} \to z^1w_1z^2\to z^2w_1z^1,\nn\\
&&\partial\bra z^1z^2w_1\ket_{21}=\bra z^1z^2\ket_{20}w_2+z^2\bra z^1w_1\ket_{11}-z^1\bra z^2w_1\ket_{11}-{\color{blue}\bra z^1w_1z^2\ket_{21}},\nn\\
&&\bra z^2z^3w_3\ket_{21}\stackrel{\opn{red}^{L'}}{\to} z^3z^2w_3\to z^3w_3z^2\to -(z^1w_1+z^2w_2)z^2\to -z^2w_1z^1-z^2w_2z^2,\nn\\
&&\hspace{1.75cm}\stackrel{\opn{red}^R}{\to}-z^2(w_1z^1+w_2z^2),\nn\\
&&\partial\bra z^2z^3w_3\ket_{21}=\bra z^2z^3\ket_{20}w_3+z^3\bra z^2w_3\ket_{11}+\bra z^3w_3\ket_{22}z_2-{\color{blue}\bra z^1w_1z^2\ket_{21}}-z^2\bra z^3w_3\ket_{11}.\nn\eea
We see that the tip $\bra z^1w_1z^2\ket_{21}$ gets hit twice, while the differential of the difference $\bra z^1z^2w_1-z^2z^3w_3\ket_{21}$ reproduces $\partial(\SF{T}_3^1)_2$.


\section{Maps between the Bi-module Resolutions}
That there should be chain equivalences between any two projective resolutions of the same module is standard, however, what will be crucial is that we will juggle the bar, minimal and Chouhy-Solotar resolutions together to navigate the otherwise intractable deformation quantisation. Therefore we spend some words spelling out the maps.
\begin{proposition}
Between the Chouhy-Solotar resolution $\P_{\sbullet}\to\Gl$ and the bar resolution $\B_{\sbullet}\to \Gl$, there are bi-module chains maps
\bea F_{\sbullet}:~\P_{\sbullet} \rightleftharpoons \B_{\sbullet}~:G_{\sbullet}.\nn\eea
There is a bi-module chain homotopy $h_n:\,\B_n\to \B_{n+1}$ satisfying
\bea &&G_{\sbullet} F_{\sbullet}=id,~~F_{\sbullet}G_{\sbullet}=id-\{\partial,h\}\label{chain_htpy}\\
&&h_{\sbullet}F_{\sbullet}=G_{\sbullet+1}h_{\sbullet}=h_{\sbullet+1}h_{\sbullet}=0.\label{chain_htpy_extra}\eea
\end{proposition}
The proofs are found in sec.5 of \cite{BarmeierWang}. We will content ourselves with describing these maps in low degrees explicitly and some examples.

First at degree 0, we have $\P_0=\B_0=\Gl\otimes\Gl$, so we set $F_0=G_0=id$ and $h_0=0$. At degree 1, $\P_1=\Gl\otimes Q_1\otimes\Gl$, and we set
\bea F_1(a\otimes x\otimes b)=a[x]b.\nn\eea
The map $G_1$ in the other way is a split: take $a[x_1\cdots x_n]b$ with $x_1\cdots x_n$ irreducible and $x_i$ being one letter words, then
\bea G_1(a[x_1\cdots x_n]b)&=&a\opn{spl}(x_1\cdots x_n)b\nn\\
&=&\overline{ax_1\cdots x_{n-1}}\otimes x_n\otimes b+\overline{ax_1\cdots x_{n-2}}\otimes x_{n-1}\otimes\overline{x_nb}+\cdots+a\otimes x_1\otimes\overline{x_2\cdots x_nb}.\nn\eea
It is easy to verify \eqref{chain_htpy}, if one sets $h_1$ to be
\bea &&h_1(a[x_1\cdots x_n]b)=-a\opn{spl}'(x_1\cdots x_n)b,\nn\\
&&\opn{spl}'(x_1\cdots x_n):=[x_1|x_2\cdots x_n]+x_1[x_2|x_3\cdots x_n]+\cdots+x_1\cdots x_{n-2}[x_{n-1}|x_n].\label{split_prime}\eea
Note the map $\opn{spl}'$ here is defined for $\B_{\sbullet}$ while $\opn{spl}$ above is for $\P_{\sbullet}$.

At degree 2, consider $a\otimes s\otimes b$ with $s\mapsto f_s$ being part of the reduction system, we set
\bea F_2(a\otimes s\otimes b)=a\opn{spl}'(s-f_s)b\label{F_2}\eea
For the more important $G_2$ map, given $[u|v]$ with $u,v$ irreducible, then $G_2([u|v])$ gives the collection of tips of relations used in applying $\opn{red}^R$ to $uv$.
\begin{example}
  Continuing from sec.\ref{sec_RfT}. Consider $[w_3w_1|z^3w_2]_{02}$, the right reduction goes as
\bea\opn{red}^R:~~w_3\udl{w_1z^3w_2}\stackrel{\eqref{rel_wzw}}{\to} \udl{w_3w_2z^3}w_1\stackrel{\eqref{new_rel}}{\to}-w_2w_1z^1w_1-w_2w_2z^2w_1.\nn\eea
The underlined terms are the tips picked up in the process.
Therefore $G_2[w_3w_1|z^3w_2]_{02}=w_3\otimes w_1z^3w_2\otimes e_2+e_0\otimes w_3w_2z^3\otimes w_1$.

Mapping this back to $\B_2$ applies the split \eqref{F_2}
\bea F_2G_2[w_3w_1|z^3w_2]&=&
w_3\opn{spl}'(w_{[1}z^3w_{2]})+\opn{spl}'(w_3w_2z^3+w_2w_2z^2+w_2w_1z^1)w_1\nn\\
&=&w_3[w_{[1}|z^3w_{2]}]+w_3w_{[1}[z^3|w_{2]}]+[w_3|w_2z^3]w_1+w_3[w_2|z^3]w_1\nn\\
&&+[w_2|w_2z^2]w_1+w_2[w_2|z^2]w_1+[w_2|w_1z^1]w_1+w_2[w_1|z^1]w_1.\nn\eea
\end{example}
The homotopy $h_2$, which will be used in the homotopy transfer later, has two parts
\bea h_2(a[x|y]b)=\sum_{\ell(d_i)>0} \underbrace{ac_i[\opn{spl}'(s_i-f_i)|d_i]b}_{h_2'}-a[\opn{spl}'(x)|y]b.\label{h_2}\eea
The second part is quite self-explanatory: one splits $x$ as in \eqref{split_prime} and stacks it to $y$.
To explain the first part, assume that in right reducing $xy$, we have applied the following relations $c_is_id_i\to c_if_id_i$. Then for each $i$ with the $d_i$ of non-zero length $\ell(d_i)>0$, we apply $\opn{spl}'(s_i-f_i)$ and stack it to $d_i$. One can check \eqref{chain_htpy_extra} i.e. $h_2F_2=0=G_2h_1$.
\begin{example}
  Take again $[w_3w_1|z^3w_2]_{02}$, we saw above that the following relations were used in the $\opn{red}^R$
  \bea w_3(w_1z^3w_2)\to w_3(w_2z^3w_1),~~(w_3w_2z^3)w_1\to -(w_2w_2z^2+w_2w_1z^1)w_1.\nn\eea
  The first item will not contribute to $h_2'$, so
  \bea h_2([w_3w_1|z^3w_2])&=&[w_3|w_2z^3|w_1]+w_3[w_2|z^3|w_1]+[w_2|w_2z^2|w_1]\nn\\
  &&+w_2[w_2|z^2|w_1]+[w_2|w_1z^1|w_1]+w_2[w_1|z^1|w_1]-[w_3|w_1|z^3w_2].\nn\eea
  The reader may check that \eqref{chain_htpy} holds here.
\end{example}

We need a recursive formula for $G_2$ and $h_2'$ whose significance will become clear later.
\begin{lemma}\label{lem_G2}(see lem 5.16 in \cite{BarmeierWang})
  Let $u$ and $v$ be two irreducible words. Write $u$ as two parts $u=u_1u_2$, then
  \bea G_2[u|v]=u_1G_2[u_2|v]+G_2[u_1|\overline{u_2v}].\label{used_VIII}\eea
  Remember that a bar on $u_2v$ stresses that $u_2v$ be reduced, even though this should be tacit.
  Further if $uv=asb$ with $s$ being the rightmost tip appearing in $uv$, then
  \bea G_2[u|v]=a\otimes s\otimes b+aG_2[f_s|b]+G_2[a|\overline{f_sb}].\label{carry_on}\eea
  Similarly for $h_2'$ (the first term of $h_2$ in \eqref{h_2}), we have
 \bea
 &&h'_2[u|v]=u_1h'_2[u_2|v]+h_2'[u_1|\overline{u_2v}],\label{h_2_recursive}\\
 &&h'_2[u|v]=ah'_2[\opn{spl}'(s-f_s)|b]+ah_2'[f_s|b]+h_2'[a|\overline{f_sb}].\nn\eea
\end{lemma}
\begin{proof}(Sketch)
We keep in mind the explicit description of $G_2$ as collecting tips of the relations used in right reducing $uv$. If $u_2v$ is reducible, then after having reduced it to $\overline{u_2v}$, one goes on with the reduction of $u_1\overline{u_2v}$. This gives the first statement.

For \eqref{carry_on}, after encountering in $uv$ the first tip $s$, which gets mapped to $1\otimes s\otimes 1$, one replaces $asb\to af_sb$ and continues with the remaining reduction. This is the rhs of \eqref{carry_on} in words.

The same argument also gives the statement about $h_2'$, as $h_2'$ and $G_2$ are both based on the right reduction procedure.
\end{proof}

Now we give the formula for $F_3$.
\begin{lemma}
Let $uvw\in \SF{A}_2$ be an ambiguity: $u$ has one letter, $uv=s$, $vw=t$ with reductions $s\mapsto f_s$ and $t\mapsto f_t$, then
\bea F_3(1\otimes uvw\otimes 1)=-[\opn{spl}'(s)|w]-h_2[f_s|w]+h_2[u|f_t]\label{F_3}.\eea
\end{lemma}
\begin{proof}
  Since we will construct $F_3$ such that $\partial F_3=F_2\partial$, we can write $F_3=k_2F_2\partial$, where $k_2$ is the chain-homotopy of the bar complex \eqref{bar_k}.
  For an ambiguity $uvw$, we reviewed in \eqref{partial_3_P} that $\partial(1\otimes uvw\otimes 1)$ is computed by collecting the tips arising from $\opn{red}^{L'}$ minus those from $\opn{red}^R$. After that $F_2$ sends a tip $a\otimes s\otimes b$ to $a[\opn{spl}'(s-f_s)]b$, and $k_2$ moves $b$ into the bar, provided $b$ is not of length zero.

  Putting these words into practice, the first reduction in $\opn{red}^{L'}$ starts from the left, picking up the tip $1\otimes s\otimes w$. The remaining $k_2F_2$ action turns it into $-[\opn{spl}'(s-f_s)|w]$. We apply afterwards right reduction to $f_sw$ and applies $k_2F_2$, which can be written as $-h'_2[f_s|w]$.

  %
  The reduction $\opn{red}^R$ is a straightforward right reduction. But do note that its first reduction gives $u\otimes t\otimes 1$, which gives zero after applying $k_2F_2$ due to the peculiar way $k_2$ works.
  Altogether
  \bea F_3(1\otimes uvw\otimes 1)=-[\opn{spl}'(s-f_s)|w]-h'_2[f_s|w]+h'_2[u|f_t]=-[\opn{spl}'(s)|w]-h_2[f_s|w]+h_2[u|f_t].\nn\eea
\end{proof}

\subsection{The Homotopy Transfer}
We use the notation $\P^{\sbullet}=\hom_{\Gl^e}(\P_{\sbullet},\Gl)$, and $\B^{\sbullet}=\hom_{\Gl^e}(\B_{\sbullet},\Gl)$ for the cochain complex and $\gd, F^{\sbullet},G^{\sbullet}$ and $h^{\sbullet}$ for the induced maps on the cochains
\bea &G^{\sbullet}:~\P^{\sbullet}\rightleftharpoons \B^{\sbullet}:~F^{\sbullet},\nn\\
&F^{\sbullet}G^{\sbullet}=id,~~~G^{\sbullet}F^{\sbullet}=id-\{\gd,h\},\nn\\
&F^{\sbullet}h^{\sbullet+1}=h^{\sbullet}G^{\sbullet}=h^{\sbullet}h^{\sbullet+1}=0.\nn\eea
As we have reviewed in sec.\ref{sec_BR}, the cochain complex $\B^{\sbullet}$ has a dg-Lie-algebra structure. We can transfer this structure to an $L_{\infty}$-structure on $\P^{\sbullet}$ (see sec 7.2 \cite{BarmeierWang}). The 2-ary Lie bracket leads to a host of higher brackets
\bea [\underbrace{\sbullet,\cdots,\sbullet}_n]:~\P^{k_1}\otimes\cdots\otimes \P^{k_n}\to \P^{k_1+\cdots+k_n-2n+3}\nn\eea
satisfying some generalised Jacobi-identity.
The formula for the higher brackets are expressed as trees and there is an $L_{\infty}$-quasi-isomorphism that intertwines the induced $L_{\infty}$-structure on $\P^{\sbullet}$ and the dg-Lie structure on $\B^{\sbullet}$.
In the following we shall focus on $\P^2$ and $\B^2$.
Given $g\in\P^2$, we need to solve the MC equation
\bea \gd g-\sum_{n=2}^{\infty}\frac{1}{n!}[\underbrace{g,\cdots,g}_n]=0.\label{MC_L_infty}\eea
This is the generalisation of the MC-equation of a dg-Lie algebra to an $L_{\infty}$-algebra, see sec.4 of \cite{Getzler_Lie}).
We introduce the $L_{\infty}$-map that will rewrite this equation in $\B^{\sbullet}$.


We use a \begin{tikzpicture}[scale=.8]
\draw [blue,-] (0,0) -- (0,0.2) -- (-0.15,0.35);
\draw [blue,-] (0,0.2) -- (0.15,0.35);
\end{tikzpicture} for the Lie bracket $[-,-]$ in $\B^{\sbullet}$ and a thick dot to mean the homotopy $h^3:\,\B^3\to \B^2$.
\bea
\begin{tikzpicture}[scale=1]
\draw [blue,-] (0,0) -- (0,0.2) -- (-0.15,0.35);
\draw [blue,-] (0,0.2) -- (0.15,0.35);
\end{tikzpicture}=[-,-],~~~~
\begin{tikzpicture}[scale=1]
\draw [blue,-] (0,0) -- (0,0.2) -- (-0.15,0.35);
\draw [blue,-] (0,0.2) -- (0.15,0.35);
\draw [fill=black] (0,0.1) circle (0.03 cm);
\end{tikzpicture}=h^3[-,-].\nn\eea
Let now $g\in\P^2$, we use \begin{tikzpicture}[scale=0.8]
\draw [blue,-] (0,0) -- (0,0.4);
\end{tikzpicture} to denote $G^2g$. With these notations, the $L_{\infty}$ map sends $g$ to
\bea g\mapsto T(g)=\begin{tikzpicture}[scale=1]
\draw [blue,-] (-.5,0) -- (-0.5,0.3);
\end{tikzpicture}+\frac12
\begin{tikzpicture}[scale=1]
\draw [blue,-] (0,0) -- (0,0.2) -- (-0.15,0.35);
\draw [blue,-] (0,0.2) -- (0.15,0.35);
\draw [fill=black] (0,0.1) circle (0.02 cm);
\end{tikzpicture}
+\frac12
\begin{tikzpicture}[scale=1]
\draw [blue,-] (0,0) -- (0,0.2) -- (-0.3,0.5);
\draw [blue,-] (-0.15,0.35) -- (-0.05,0.5);
\draw [blue,-] (0,0.2) -- (0.15,0.35);
\draw [fill=black] (0,0.1) circle (0.02 cm);
\draw [fill=black] (-0.07,0.275) circle (0.02 cm);
\end{tikzpicture}+\frac18\begin{tikzpicture}[scale=1]
\draw [blue,-] (0,0) -- (0,0.2) -- (-0.15,0.35) -- (-0.3,0.5);
\draw [blue,-] (-0.15,0.35) -- (-0.05,0.5);
\draw [blue,-] (0,0.2) -- (0.3,0.5);
\draw [blue,-] (0.15,0.35) -- (0.12,0.5);
\draw [fill=black] (0,0.1) circle (0.02 cm);
\draw [fill=black] (-0.07,0.275) circle (0.02 cm);
\draw [fill=black] (0.07,0.275) circle (0.02 cm);
\end{tikzpicture}+\cdots \in\B^2\label{tree_formula}\eea
The sum runs over all the binary planar trees. Using the $T$ map, \eqref{MC_L_infty} reads
\bea \gd g-\frac12F^3[T(g),T(g)]=0,\label{MC_L_infty_F}\eea
that is, the sum involving higher brackets in \eqref{MC_L_infty} is induced from $[T(g),T(g)]$.
The $T$ map intertwines the MC-equation \eqref{MC_L_infty} with the MC-equation in $\B^{\sbullet}$
\bea \gd T(g)-\frac12[T(g),T(g)]=0\label{MC_B}.\eea
That $T$ is an $L_{\infty}$-quasi-isomorphism means that \eqref{MC_B} and \eqref{MC_L_infty_F} mutually imply. This will be crucial as it is far easier checking \eqref{MC_L_infty_F}.

\begin{remark}
\begin{enumerate}
  \item We find it more useful to read the tree formula \eqref{tree_formula} from the root up. Given $[u|v]\in \B_2$, plug it to the root of the tree and from the leaves out comes an element in $\Gl$. The first term in \eqref{tree_formula} applies $G_2$ and evaluates with $g$: $[u|v]\mapsto g(G_2[u|v])$. The second term uses $h_2$ to split $[u|v]$ into three parts and evaluates with $[G^2g,G^2g]\in\B^3$, i.e. $[u|v]\mapsto [G^2g,G^2g](h_2[u|v])$, and so on.
  \item The simplest way to keep track of the combinatorial factors of the trees is to remember that a tree is self-similar: at each node of the tree, what has come to pass so far is also what is going to happen. This is easily expressed also as a picture
  \bea T(g)=~\begin{tikzpicture}[scale=1]
  \draw [fill = blue,opacity=.5] (0,0) circle (0.2 cm);
  \draw [blue,-] (0,-0.4) -- (0,-0.2);
\end{tikzpicture}~=~
  \begin{tikzpicture}[scale=1]
  \draw [blue,-] (-1,0) -- (-1,0.5);
  \node at (-0.5,0.3) {$+\frac12$};
\draw [blue,-] (0,0) -- (0,0.5) -- (-0.35,0.85);
\draw [blue,-] (0,0.5) -- (0.35,0.85);
\draw [fill = blue,opacity=.5] (-0.4,.9) circle (0.2 cm);
\draw [fill = blue,opacity=.5] (0.4,.9) circle (0.2 cm);
\draw [fill=black] (0.00,0.25) circle (0.03 cm);
\end{tikzpicture}\label{recursive_tree}\eea
that is, each blob is defined as the sum of the straight line and a \begin{tikzpicture}[scale=.8]
\draw [blue,-] (0,0) -- (0,0.2) -- (-0.15,0.35);
\draw [blue,-] (0,0.2) -- (0.15,0.35);
\end{tikzpicture} connected to two blobs, and this persists ad infinitum.
\item We also assume tacitly that $g$ has some built-in small parameter or formal parameter $\hbar$, and the outcome $T(g)$ is either a convergent power series or a formal series.
  \end{enumerate}
\end{remark}

\subsection{Combinatorial Deformation Quantisation}\label{sec_CDQ}
Recall that $\P_2$ is generated by the tips in the reduction system. Let $g\in\P^2$ be a 2-cochain that sends $1\otimes s\otimes 1$ (where $s$ is a tip) to $g_s$. Then $m+T(g)\in\B^2$ gives a deformation of $m$, whose associativity will be determined by the MC-equation \eqref{MC_B}.
But before discussing this, we give a combinatorial description of this deformed product.
In fact we state a theorem that is the technical underpinning of the current work
\begin{theorem}\label{thm_combo_star}(Theorem 7.21 in \cite{BarmeierWang})
Given any $g\in\P^2$, we can define a multiplication rule as
\bea u\star v:=\overline{uv}+T(g)([u|v]).\label{star}\eea
This multiplication is not associative for now, but can be described as applying the right reduction with the deformed reduction system
\bea s\mapsto f_s+g_s.\label{deform_reduction}\eea
to $uv$ until it becomes irreducible.
\end{theorem}
\begin{proof}(Sketch)
In the proof, a bar placed on a word e.g. $\overline{u_1\cdots u_n}$ means the word is reduced. Note that the undeformed reduction system is such that one can reduce a word in any way but have the same outcome. We also use the shorthand $T(g)\to T$.

We first prove that if $uv$ is also irreducible then $T([u|v])=0$. Indeed,
plug $[u|v]$ to the root of the tree in fig.\ref{recursive_tree}. The term \begin{tikzpicture}[scale=1]
\draw [blue,-,thick] (0,0) -- (0,0.3);
\end{tikzpicture} gives the evaluation $gG_2[u|v]$, which is zero by the assumed irreducibility.
Next one applies $h_2$ to $[u|v]$ and get $-[\opn{spl}'(u)|v]$ only. We write $[\opn{spl}'(u)|v]=\sum_ix_i[y_i|z_i|v]$.
The recursion says that we evaluate $T[y_i|z_i]$ and $T[z_i|v]$. For both terms, we iterate the same argument except now that $y_i,z_i$ have shorter length and eventually $\opn{spl}'$ will vanish on them.

Now assume $uv$ is reducible but $u$ has only one letter. We write $v=v_1\cdots v_n$ and $uv=sb$ with $s$ the rightmost tip appearing in $uv$
\bea uv=\underbrace{uv_1\cdots v_i}_s\underbrace{v_{i+1}\cdots v_n}_b\nn\eea
Note $b$ is irreducible here.
To evaluate $gG_2[u|v]$, we use \eqref{carry_on} and get $\overline{g_sb}+gG_2[f_s|b]$.
The remaining term in $T$ gives
\bea \frac12[T,T](h_2[u|v])=\frac12[T,T]\big([\opn{spl}'(s-f_s)|b]+h_2'[f_s|b]\big)
=\frac12[T,T]\big([\opn{spl}'(s)|b]\big)+\frac12[T,T](h_2[f_s|b]).\nn\eea
We recall from \eqref{Lie_bkt_bar} that
\bea \frac12[T,T]=T\circ_1T-T\circ_2T,\nn\eea
leading to
\bea \frac12[T,T]([\opn{spl}'(s)|b])=T\circ_1T([\opn{spl}'(s)|b])=T([g_s|b]).\nn\eea
To see the last two equalities, we spell out $[\opn{spl}'(s)|b]$ as
\bea [\opn{spl}'(s)|b]=[u|v_1\cdots v_i|b]+u[v_1|v_2\cdots v_i|b]+\cdots+\overline{uv_1\cdots v_{i-2}}[v_{i-1}|v_i|b].\nn\eea
Evaluating $[v_j\cdots v_i|b]$ on $T$ vanishes since $v_j\cdots v_ib$ is irreducible, hence $T\circ_2T[\opn{spl}'(s)|b]=0$.
As for $T\circ_1T[\opn{spl}'(s)|b]$, we need to
evaluate $T[\opn{spl}'(s)]$. This is only nonzero for the summand $[u|v_1\cdots v_i]$, with the others vanishing due to irreducibility. We have $T[u|v_1\cdots v_i]=gG_2[u|v_1\cdots v_i]=g_s$. 
Putting these together
\bea u\star v&=&\overline{uv}+\overline{g_sb}+gG_2[f_s|b]+T([g_s|b])+\frac12[T,T](h_2[f_s|b])\nn\\
&=&\overline{f_sb}+\overline{g_sb}+T([f_s+g_s|b])=(f_s+g_s)\star b.\nn\eea
This formula says when multiplying $u$ to $v$, one reduces the first encountered tip $s$ to $f_s+g_s$ and starts over with $(f_s+g_s)\star b$.
The next step will show that for $u$ of length larger than 1, one can always reduce it to the length 1 case.
In particular, to compute $(f_s+g_s)\star b$, one splits $f_s+g_s$ into one letter words and multiply them to $b$ one by one. This would finish the proof.

Write $u$ as $u_1u_2$ with $u_1$ having one letter
\bea (u_1u_2)\star v&=&\overline{u_1u_2v}+\underbrace{\overline{u_1gG_2([u_2|v])}+gG_2([u_1|\overline{u_2v}])}_{\eqref{used_VIII}}+\frac12[T,T](\underbrace{u_1h'_2[u_2|v]+h_2'[u_1|\overline{u_2v}]}_{\eqref{h_2_recursive}}-[\opn{spl}'(u_1u_2)|v]\big)\nn\\
&=&\overline{u_1u_2v}+\overline{u_1T[u_2|v]}+gG_2([u_1|\overline{u_2v}])+\frac12[T,T](h_2'[u_1|\overline{u_2v}]-[u_1|u_2|v]\big)\nn\\
&=&\overline{u_1u_2v}+\overline{u_1T[u_2|v]}+T([u_1|\overline{u_2v}])+\frac12[T,T](-[u_1|u_2|v]\big)\nn\eea
where we used $[\opn{spl}'(u_1u_2)|v]=u_1[\opn{spl}'(u_2)|v]+[u_1|u_2|v]$. Note that
\bea \frac12[T,T](-[u_1|u_2|v]\big)=T\circ_2T([u_1|u_2|v])=T([u_1|u_2\star v-\overline{u_2v}]),\nn\eea
therefore
\bea (u_1u_2)\star v&=&\overline{u_1u_2v}+\overline{u_1T[u_2|v]}+T([u_1|\overline{u_2v}])+T([u_1|u_2\star v])-T([u_1|\overline{u_2v}])\nn\\
&=&\overline{u_1(u_2\star v)}+T([u_1|u_2\star v])=u_1\star(u_2\star v).\nn\eea
Using this procedure, one can successively reduce the length of $u$ until it is of length 1.
\end{proof}
As we have reviewed, if $T(g)$ satisfies the MC-equation \eqref{MC_B}, then $\star$ would be associative. Solving the MC-equation is a major undertaking and the next theorem lets us evade it.
Recall that the degree 3 term of the Chouhy-Solotar resolution \eqref{non_min_res} is generated by the ambiguities $\SF{A}_2$ (see sec.\ref{sec_RS}). If $\star$ were to be associative, then to say the least, for all ambiguities $uvw$, we must have
\bea (u\star v)\star w=u\star (v\star w)\label{least}.\eea
\begin{theorem}\label{thm_ass_chk}(Theorem 7.37 in \cite{BarmeierWang})
The condition \eqref{least} is also sufficient, that is, the star product \eqref{star} is associative iff \eqref{least} holds for all the ambiguities generating $\SF{A}_2$.
\end{theorem}
This result is somewhat expected: since the ambiguities are what potentially make our reduction based multiplication rule dependent on parenthesising, but if associativity can be secured on all ambiguities, we are through.
\begin{proof}(Sketch)
It is crucial here that the $L_{\infty}$-map $T$ be invertible and so $g\in \P^2$ satisfies \eqref{MC_L_infty_F} iff $T(g)\in \B^2$ satisfies \eqref{MC_B}. Therefore we need only check the former on all generators of $\P_3=\SF{A}_2$. Let $uvw$ be an ambiguity,
by def.\ref{def_ambiguity}, $u$ has one letter and $uv=s$ is a tip and $vw=t$ is also a tip.
The tips have reduction $s\mapsto f_s$ and $t\mapsto f_t$, while the cochain $g$ has $g(1\otimes s\otimes 1)=g_s$ and $g(1\otimes t\otimes 1)=g_t$.


We need to check
\bea (\gd g-\frac12F^3[T,T])(1\otimes uvw\otimes1)\stackrel{?}{=}0\label{question}.\eea
The first term gives $g(1\otimes s\otimes w-u\otimes t\otimes1)+gG_2[f_s|w]-gG_2[u|f_t]$ according to the description of $\partial_3$ in \eqref{partial_3_P}. Thus
\bea (\gd g)(1\otimes uvw\otimes1)=\overline{g_sw}-\overline{ug_t}+gG_2[f_s|w]-gG_2[u|f_t].\nn\eea
For the remaining terms in \eqref{question}, we use \eqref{F_3} for $F_3$
\bea -F_3(1\otimes uvw\otimes1)=[\opn{spl}'(s)|w]+h_2[f_s|w]-h_2[u|f_t]\nn\eea
We compute each term
\bea
\frac12[T,T]([\opn{spl}'(s)|w])&=&(T\circ_1T-T\circ_2T)[\opn{spl}'(s)|w]\nn\\
&=&T\circ_1T([u|v|w]+\hcancel[blue]{u[\opn{spl}'(v)|w]})-T\circ_2T(([u|v|w]+\hcancel[blue]{u[\opn{spl}'(v)|w]})\nn\\
&=&T[g_s|w]-T[u|g_t],\label{used_VII}\\
\frac12[T,T](h_2[f_s|w]-h_2[u|f_t])&\stackrel{\eqref{recursive_tree}}{=}&T[f_s|w]-gG_2[f_s|w]-T[u|f_t]+gG_2[u|f_t]\nn\eea
where the crossed out terms are due to $T[a|b]=0$ if $a,b,ab$ are all irreducible, proved in thm.\ref{thm_combo_star}.
So altogether we get
\bea  MC(g)(1\otimes uvw\otimes1)&=&\overline{g_sw}-\overline{ug_t}+gG_2[f_s|w]-gG_2[u|f_t]\nn\\
&&+T[g_s|w]-T[u|g_t]+T[f_s|w]-gG_2[f_s|w]-T[u|f_t]+gG_2[u|f_t]\nn\\
&=&\overline{(f_s+g_s)w}-\overline{u(f_t+g_t)}+T[f_s+g_s|w]-T[u|f_t+g_t]\nn\\
&=&(f_s+g_s)\star w-u\star(f_t+g_t)\nn\\
&=&(u\star v)\star w-u\star(v\star w).\nn\eea
%
%
This finishes the proof.
\end{proof}

\section{Warm up examples}
\subsection{$A_1$}
We now deal with the $A_1$-singularity example. Recall that we eventually aim at the deformed singularity
\bea X_{r^2}=\{(x_1,x_2,x_3)\in\BB{C}^3|f_{r^2}(x)=x_1^2+x_2^2+x_3^2-r^2=0\}.\nn\eea
For generic $r$, we call the real locus $M_{r^2}$, and the quantisation via branes procedure is applied to the embedding $M_{r^2}\hookrightarrow X_{r^2}$.
Note that one can take the real locus in different ways, e.g. after a change of variables one can write $f(x)=x_+x_-+4x_3^2$ with $x_{\pm}=(x_1\pm ix_2)/2$.
Setting $x_{\pm},x_3$ real leads to a different $M'$ (non-compact). This is extensively discussed in \cite{Gukov:2008ve}, but we shall focus on compact $M$.

We start however not with $f_{r^2}$ but $f_0$, and obtain a quiver algebra which is a non-commutative crepant resolution of $f_0$. We will calculate $\HH^2$ of this algebra and carry out deformation quantisation. We get two deformation parameters, one of which would correspond to a Poisson structure, another would be a Kodaira-Spencer type deformation and go to the $r^2$ parameter in the deformation above (i.e. we get to the deformed singularity by way of the resolved singularity).
To summarise the result, for arbitrary parameters $q,t$, the polynomial algebra of $x_{\pm},x_3$ are deformed into
\bea &&x_+\star x_-=x_-\star x_++(q-q^{-1})x_3,\nn\\
&& x_3\star x_+=x_+\star x_3+2(q-q^{-1})x_+\nn\\
&&x_3\star x_-=x_-\star x_3-2(q-q^{-1})x_-,\nn\\
&&2(x_+\star x_-+x_-\star x_+)+x_3\star x_3=(t-q)(t+q-2q^{-1}).\label{A_1_deform_tq}\eea
Moreover, if $t$ satisfies the integrality condition
\bea \frac{t-q}{q-q^{-1}}=-(n+2),\nn\eea
then we recover the geometric quantisation on $(S^2,\go=\re\Go)$ with $\go/2\pi$ having period $n$.

\vskip 0.5cm

We start with the combinatorial model of $D^b(X)$ as a quiver algebra in ex.\ref{Ex_A1}. The variables $x_{1,2,3}$ are related to $z^{1,2},w_{1,2}$ as
\bea (x_1,x_2,x_3)=[w_1,w_2]\gs^{1,2,3}
\begin{bmatrix}
  z^1 \\ z^2
\end{bmatrix}=(w_1z^2+w_2z^1,iw_2z^1-iw_1z^2,w_1z^1-w_2z^2)\label{xyz}\eea
where $\gs^{1,2,3}$ are the three Pauli-matrices. The relation $w_iz^i=0$ implies $f_0=0$.

As a first step, we will consider the Hochschild cohomology $\HH^2(\Gl)$ to capture the infinitesimal deformations.
Though it is easily computable from the HKR map \eqref{HKR}, we will need explicit representatives of $\HH^2(\Gl)$ at the cochain level.
\begin{proposition}
  The Hochschild cohomology is graded by the path length, at grading $-2$, we have the explicit cocycles
  \bea &g(e_1\otimes z^2w_2\otimes e_1)=t_1e_1,~~~~g(e_0\otimes w_2z^2\otimes e_0)=t_0e_0,\nn\\
  &g(e_0\otimes w_1z^1w_2\otimes e_1)=(t_0-t_1)e_0w_2e_1,~~~g(e_1\otimes z^1w_1z^2\otimes e_0)=(t_1-t_0)e_1z^2e_0.\nn\eea
  The second line is determined from the first, since the relations appearing in the second line are redundant.
\end{proposition}
\begin{proof}
To compute the Hochschild cohomology, we can either use the minimal resolution resulting from the optimal ordering \eqref{order_best}, or the sub-optimal one \eqref{order_bad}. We use the latter for illustrative purposes.

Let $g\in \HH^2$ be 2-cochain. At grading $-2$, i.e. $g$ reduces the path lengths by 2, one can easily check that the assignment
\bea g(e_0\otimes w_2z^2\otimes e_0)=t_0,~~g(e_1\otimes z^2w_2\otimes e_1)=t_1,\label{used_I}\eea
would lead to a cocycle. Besides, keeping in mind that the tips $e_1z^1w_1z^2e_0$, $e_0w_1z^1w_2e_1$ get mapped onto isomorphically from $e_1z^2w_2z^2e_0$ and $e_0w_2z^2w_2e_1$, which live in $\P^3$. Thus $g$ evaluated on $e_1z^1w_1z^2e_0$, $e_0w_1z^1w_2e_1$ is completely determined from \eqref{used_I}.
It is equally clear that \eqref{used_I} cannot be a coboundary, since coboundaries cannot have grading $-2$.
\end{proof}
We can match the two cocyles to the Poisson bracket. From the HKR map, we expect the following
\bea gG_2([a|b]-[b|a])\sim \{a,b\}.\nn\eea
  On the lhs $a,b$ are paths, and on the rhs, they are identified as functions via \eqref{xyz}.
  The Poisson bracket is induced from the holomorphic symplectic form
  \bea \hspace{2cm}\Go=\frac{1}{x_3}dx_1\wedge dx_2~~\To~~\{x_1,x_2\}=x_3~~\{x_2,x_3\}=x_1,~~\{x_3,x_1\}=x_2.\label{PB}\eea
Note the Poisson bracket has $\{f_0,-\}=0$ and so descends to $X_0$.

Now compute the `commutators' $gG_2([w_2z^1|w_1z^2]-[w_1z^2|w_2z^1])$ and $gG_2([w_2z^1|w_1z^1]-[w_1z^1|w_2z^1])$, for which we need (as usual we omit the $e_i$'s when they are clear)
\bea gG_2[w_2z^1|w_1z^2]&=&-(t_1-t_0)w_1z^1+t_0w_1z^1,\nn\\
gG_2[w_1z^2|w_2z^1]&=&w_1z^1t_1\nn\\
gG_2[w_2z^1|w_1z^1]&=&0,\nn\\
gG_2[w_1z^1|w_2z^1]&=&(t_0-t_1)w_2z^1.\nn\eea
Translating these into Poisson brackets
\bea
\{x_+,x_-\}&=&(t_1-t_0)x_3\nn\\
\{x_3/2,x_-\}&=&(t_0-t_1)x_-,\nn\eea
and comparing with \eqref{PB}, we get
\bea gG_2([a|b]-[b|a])=2i(t_1-t_0)\{a,b\}.\nn\eea
%

We see so far that the difference $t_0-t_1$ controls the deformation due to Poisson structure. Another linear combination of $t_0,t_1$ would control the Kodaira-Spencer deformation. Recall that for a variety $X$ the cohomology group $H^1(X,TX)$ controls the infinitesimal deformation of the complex structure. But if $X$ is affine, defined by the vanishing of a polynomial $f_0$, then $H^1(X,TX)$ can also be regarded as deforming $f$. According to \cite{GNTjurina_1969}, such deformations are given by
\bea H^1(X,TX)\simeq\frac{\BB{C}[x_1,x_2,x_3]}{\bra f,\partial_{x_i}f\ket}.\label{match_Poisson_A1}\eea
For our $f_0$, the above quotient is $\BB{C}$, i.e. a constant deformation $f_0\to f_0+\ep$. 

Having identified the infinitesimal deformations, we use the combinatorial version of deformation quantisation in sec.\ref{sec_CDQ} to get a finite deformation.
With the criterion of thm.\ref{thm_ass_chk}, the procedure is a breeze.
We have the following deformation of the reduction system
\bea
&&z^2w_2\to -z^1w_1+(t-q^{-1})e_1,\nn\\
&&w_2z^2\to -w_1z^1+(t-q)e_0,\nn\\
&&z^1w_1z^2\to z^2w_1z^1+(q-q^{-1})z^2,\nn\\
&&w_1z^1w_2\to w_2z^1w_1-(q-q^{-1})w_2.\label{A_1_deform_q}\eea
The right reduction using the deformed reduction system leads to the $\star$ product.
Here the parameter $q$ corresponds to the Poisson structure and $t-q$ corresponds to $f_0\rightsquigarrow f_{r^2}$.
To ensure that this gives an associative algebra, we need only check the associativity on the ambiguities $z^2w_2z^2$ and $w_2z^2w_2$, and the next remark says that even this is unnecessary.
\begin{remark}
We point out that the third and fourth line in \eqref{A_1_deform_q} are completely fixed from the first two, by demanding associativity on $z^2w_2z^2$ and $w_2z^2w_2$. This is hardly surprising, seeing as these two ambiguities map isomorphically onto $z^1w_1z^2$ and $w_1z^1w_2$ (see ex.\ref{Ex_A1_III} and the remark therein).
In fact, the algebra has a length 2 minimal resolution had we used the ordering \eqref{order_best}. In such case, there is no ambiguity to start with. We will use the same reasoning with other global dimension 2 algebras such as $A_n$ and $D_4$ below.
\end{remark}

We can translate \eqref{A_1_deform_q} into a deformed multiplication written in $x_{1,2,3}$.
First \eqref{xyz} needs rewriting
\bea (x_+,x_-,x_3)=(w_1z^2,w_2z^1,2w_1z^1-t+q).\nn\eea
Then we compute the star product
\bea x_+\star x_-&=&(w_1z^2)\star(w_2z^1)=-w_1z^1w_1z^1+(t-q^{-1})w_1z^1,\nn\\
x_-\star x_+&=&(w_2z^1)\star(w_1z^2)=w_2\star(z^2w_1z^1)+w_2\star z^2(q-q^{-1})\nn\\
&=&-w_1z^1w_1z^1+(t-q)w_1z^1-w_1z^1(q-q^{-1})+(t-q)(q-q^{-1}),\nn\\
x_3\star x_+&=&2w_1z^2w_1z^1+2(q-q^{-1})x_+-(t-q)x_+,\nn\\
x_+\star x_3&=&2w_1z^2w_1z^1-(t-q)x_+,\nn\\
x_3\star x_-&=&2w_2z^1w_1z^1-2(q-q^{-1})x_--(t-q)x_-,\nn\\
x_-\star x_3&=&2w_2z^1w_1z^1-(t-q)x_-.\nn\eea
From these relations we get \eqref{A_1_deform_tq}. In the remainder of this section we omit $\star$.

\smallskip So far $q,t$ are arbitrary, but as $t-q$ corresponds to the deformation parameter in $f_{r^2}$ and hence controls the size of the real locus of $x_1^2+x_2^2+x_3^2=r^2$, which is a 2-sphere. The size here translates to the size of the K\"ahler form in the geometric quantisation and we expect some integrality to appear.

We let $\Gl_{q,t}$ be the deformed path algebra and
\bea e_i\Gl_{q,t}e_j=\{\textrm{Deformed paths from $e_j$ to }e_i\}.\nn\eea
From this identification emerges the brane picture:
\bea e_0\Gl_{q,t}e_0\sim \hom(B_{cc},B_{cc}).\nn\eea
The paths $e_0\Gl e_0$ have been identified as holomorphic functions on $X_0$ as in \eqref{xyz}, and so $e_0\Gl_{q,t}e_0$ are the $q,t$-deformed functions on $X_{r^2}$. In the language of sec.\ref{sec_S}, one regards these as strings stretched between $B_{cc}$.

To understand this algebra, we make some redefinitions
\bea E=\frac{1}{q-q^{-1}}w_1z^2,~~~F=\frac{1}{q-q^{-1}}w_2z^1,~~~H=\frac{1}{q-q^{-1}}(w_1z^1-\frac12(t-q)),\nn\eea
these satisfy the standard commutation relations of $\FR{sl}(2)$
\bea [E,F]=2H,~~[H,E]=E,~~[H,F]=-F.\label{sl(2)}\eea
In view of the relations \eqref{sl(2)}, we have the identification
\bea e_0\Gl_{q,t}e_0\sim U_{\FR{sl}(2)}\nn\eea
the universal enveloping algebra of $\FR{sl}(2)$.

Next we would like to find some interesting quotient modules, namely branes and in particular branes corresponding to the geometric quantisation of $S^2$.
We impose the following conditions on $e_0\Gl_{q,t}e_0$
\bea e_0w_1z^1=\gl e_0,~~e_0w_2z^2=(t-q-\gl)e_0,~~e_0w_2z^1=0.\nn\eea
This is only consistent if $\gl=-q+q^{-1}$:
\bea 0=e_0w_2z^1w_1=-e_0w_2z^2w_2+(t-q^{-1})e_0w_2=-(t-q-\gl)e_0w_2+(t-q^{-1})e_0w_2=(q-q^{-1}+\gl)e_0w_2\nn\eea
In a path, we can always use \eqref{sl(2)} to move the $H$ and $F$ to the left, which acts on $e_0$ with weight $\gl$ or kills it.
It would not have escaped the reader that this is just the Verma module
\bea \BB{C}_{\mu}\otimes_{B_-}U_{\FR{sl}(2)}.\label{verma_mod}\eea
Here $B_-$ is generated by $H,F$. The tensor product over $B_-$ means that $H,F$ can be moved over to act on $\BB{C}_{\mu}$, with $H$ acting with weight $\mu$ and $F$ killing it. Our Verma module looks backwards because our paths go to the left, alas. But at any rate, this is regarded as $\hom(B,B_{cc})$: strings stretched between some brane $B$ and $B_{cc}$.

To find the $B$ that corresponds to the brane for geometric quantisation, we need find some finite dimensional quotient modules, which is where the integrality comes in.
\begin{lemma}
  If we set
  \bea \frac{t-q}{q-q^{-1}}=-(n+2)\nn\eea
  then we can quotient $e_0\Gl_{q,t}e_0$ from the left by $e_0(w_1z^2)^{n+1}e_0$ and get a module of dimension $n+1$.
\end{lemma}
With hindsight, we should have used $t-q$ as the deformation parameter, then \eqref{A_1_deform_tq} would also look prettier.
Also one could do entirely the same analysis for $e_1\Gl_{q,t}e_0$.
\begin{proof}
  We need only check consistency. Since $e_0E^{n+1}$ is set to zero, we multiply $F$ to it from the right
  \bea e_0E^{n+1}F=(n+1)e_0(2H-n)E^n.\nn\eea
  Consistency requires that $2e_0H=ne_0$ and hence
  \bea \frac{1}{q-q^{-1}}(-q+q^{-1}-\frac12(t-q))=\frac{n}{2}\nn\eea
  solving this gives us the condition stated.
\end{proof}
The quotient we are taking here is familiar in the context of Verma modules. Note $H$ acts on $e_0$ with weight
\bea \mu=-1-\frac12\frac{t-q}{q-q^{-1}}.\nn\eea
For generic $\mu$, the Verma module \eqref{verma_mod} is simple. But if $\mu$ happens to be a dominant weight of $\FR{sl}(2)$, then there is a finite dimensional quotient, being none other than the highest weight representation of weight $\mu$. This is also the picture from geometric quantisation, itself being a special case of the Borel-Weil theorem.

\subsection{$A_n$}
The $A_n$ case is a rather easy generalisation of $A_1$. Start with $\{(x_1,x_2,x_3)\in\BB{C}^3|f(x)=x_1^2+x_2^2-x_3^{n+1}=0\}$, we get to its Mckay-Auslander-Reiten quiver \eqref{quiver_An}. The minimal resolution of this quiver algebra is also very easy (for quiver algebras of global dimension 2, their minimal resolution have the same form). From this resolution we will obtain $\HH^2(\Gl)$ and proceed with the deformation quantisation.

\smallskip

To start, the relations of the quiver is in \eqref{relation_A_n}
\bea R_i=(e_i\ot e_{i-1}\ot e_i)-(e_i\ot e_{i+1}\ot e_i)=0,\nn\eea
for ease of notation, we name these $R_i$. The resolution reads
\bea \P_{\sbullet}:~~\bigoplus_i\Gl e_i\otimes R_i\otimes e_i\Gl  \to \bigoplus_i(\Gl e_{i+1}\otimes e_i\Gl \oplus\Gl e_{i-1}\otimes e_i\Gl ) \to \bigoplus_i \Gl e_i\otimes e_i\Gl.\nn\eea
From the experience garnered in the $A_1$ case, we consider the explicit cocyle representing $\HH^2$
\bea g(e_i\otimes R_i\otimes e_i)=t_ie_i.\nn\eea
This leads to a deformation
\bea (e_i\ot e_{i-1}\ot e_i)=(e_i\ot e_{i+1}\ot e_i)+t_ie_i\label{relation_A_n_def}\eea
giving an associative algebra with no need for further checks.

To translate the quiver algebra back to functions on $X$, we set
\bea x_+&=&e_0\ot e_n\ot e_{n-1}\ot\cdots \ot e_1\ot e_0,\nn\\
x_-&=&e_0\ot e_1\ot \cdots \ot e_{n-1}\ot e_n\ot e_0,\nn\\
x_3&=&e_0\ot e_1 \ot e_0.\nn\eea
The deformed relations \eqref{relation_A_n_def} immediately gives
\bea &&x_3\star x_+=x_+\star x_3-\sum_0^nt_ix_+,\nn\\
&&x_3\star x_-=x_-\star x_3+\sum_0^nt_ix_-,\nn\\
&&x_-\star x_+=(x_3-\sum_1^nt_i)(x_3-\sum_1^{n-1}t_i)\cdots (x_3-t_1)x_3,\nn\\
&&x_+\star x_-=(x_3+t_0+\sum_1^nt_i)(x_3+t_0+\sum_2^nt_i)\cdots (x_3+t_0).\label{An_deform}\eea
This algebra has been studied in \cite{Smith_sl2} by Smith, who defined an algebra similar to $U\FR{sl}(2)$
\bea
&&[H,E]=E,~~[H,F]=-F,~~[E,F]=f(H),\nn\eea
for some polynomial function $f$.
Let us set $\sum_0^nt_i=t$ and identify $x_+=E$, $x_-=F$ and $x_3=H$,
then we have from \eqref{An_deform}
\bea &&[H,E]=E,~~[H,F]=-F,~~FE=v(H),~~EF=v(H+t),\nn\\
&&v(x)=(x-\sum_1^nt_i)(x-\sum_1^{n-1}t_i)\cdots (x-t_1)x.\nn\eea
Thus we can set $f(x)=v(x+t)-v(x)$. This result has been pointed out in \cite{HODGES1993271}, where various properties of this algebra have also been studied.

We now repeat the reasoning for the $A_1$ case and identify some branes. Consider again $e_0\Gl_{\vec t}e_0$, and we would like to impose
\bea e_0H=\gl e_0,~~~e_0F=0,~~~e_0E^N=0.\nn\eea
The consistency requires
\bea &&0=e_0F~\To~0=e_0FE=e_0v(H),\nn\\
&&0=e_0E^N~\To~0=e_0E^NF=e_0E^{N-1}v(H+t)=e_0v(H+Nt)E^{N-1}\nn\eea
This shows that both $\gl$ and $\gl+Nt$ must be roots of $v(x)$.

\subsection{$D_4$}\label{sec_D4}
The method for studying the deformation of $A_n$-singularity in \cite{HODGES1993271} gets impractical for other Kleinian singularities. In \cite{WCBMPH} another simpler approach is used in that one directly deforms the skew group algebra. Here we show that this deformation also arises from deformation quantisation using $D_4$ as example.

The $D_4$ singularity is given by $X=\{(u,v,w)\in\BB{C}^3|f=v^2-u^2w+4w^3=0\}$. The Mckay-Auslander-Reiten quiver is given in \eqref{D4_quiver},
 \bea
   \begin{tikzpicture}
   \matrix (m) [matrix of math nodes, row sep=3em, column sep=4em]
     {    & e_0 & \\
          e_1 & e_c & e_3 \\
            & e_2 &  \\ };
\path[->,bend left=10]
   (m-2-1) edge (m-2-2)
   (m-2-3) edge (m-2-2)
   (m-1-2) edge (m-2-2)
   (m-3-2) edge (m-2-2);
\path[->,bend left=10]
   (m-2-2) edge (m-2-1)
   (m-2-2) edge (m-2-3)
   (m-2-2) edge (m-1-2)
   (m-2-2) edge (m-3-2);

   \node at (5.5,0) {$\begin{array}{cc}
    (e_i\ot e_c\ot e_i)=0,~~~i=0,1,2,3  \\
    \sum_i(e_c\ot e_i\ot e_c)=0 \\
  \end{array}$};
\end{tikzpicture}\label{rel_D4}\eea
We deform the relations to
\bea &&(e_i\ot e_c\ot e_i)=t_ie_i,~~~i=0,1,2,3,\nn\\
&&\sum_i(e_c\ot e_i\ot e_c)=se_c.\nn\eea
We translate the quiver algebra back to the $u,v,w$ variables. Consider the following paths starting and ending at $e_0$
\bea &&z_1=(e_0\ot e_c \ot e_1 \ot e_c \ot e_0),~~~z_2=(e_0\ot e_c \ot e_2 \ot e_c \ot e_0),\nn\\
&&z_{21}=(e_0 \ot e_c \ot e_2 \ot e_c \ot e_1 \ot e_c \ot e_0).\nn\eea
By following the maps given in \eqref{D4_quiver}, \eqref{D4_maps}, we get their relation to $u,v,w$
\bea & w=z_1,~~u=-4z_2-2z_1,~~v=4z_{21},\nn\\
& v^2-u^2w+4w^3=0~\Leftrightarrow~z_{21}z_{21}=z_2z_1z_1+z_2z_2z_1.\nn\eea
After deformation, these relations turn into
\bea
  z_1\star z_2&=&z_2\star z_1-2s_tz_{21}-s_t(s_{03}s_0t_0-(s_{01}+s_{03})z_1-(s_{02}+s_{03})z_2),\nn\\
  z_1\star z_{21}&=&z_{21}\star z_1-s_t(z_1\star z_1+2z_2\star z_1-(s_{12}+s_{13})z_{21})\nn\\
&&-\big(-(s_{12}+s_{13}+s_{03})s_0t_0+s_{13}s_1t_1+(s_{12}+s_{13})(s_{01}s_{13})\big)z_1,\nn\\
z_2\star z_{21}&=&z_{21}z_2+s_tz_2\star z_2+s_t(2z_2\star z_1-(s_{12}+s_{23})z_{21}+(-s_0t_0+s_{23}s_{02})z_2),\nn\\
z_{21}\star z_{21}&=&z_2\star z_1\star z_1+z_2\star z_2\star z_1-(s_{23}+s_{12})z_{21}\star z_1-(s_{12}+s_{13})z_{21}\star z_2\nn\\
&&+(-s_t(s_{12}+s_{13})+s_{02}s_{23}-s_0t_0)z_2\star z_1-s_t(s_{12}+s_{13})z_2\star z_2\nn\\
&&+\big((s_{12}+s_{13})(s_2t_2+s_{02}s_{12}+s_ts_{23})-s_{13}s_1t_1\big)z_{21}\nn\\
&&-s_t(s_{12}+s_{13})(-s_0t_0+s_{23}s_{02})z_2,\label{hairy_algebra}\eea
where $s_{i\cdots k}:=s-t_i-\cdots-t_k$ and $s_t=s_{01}+s_{23}$. It is not recommended that the reader check the associativity by hand, writing a few lines of codes could be helpful.

In closing, we hope to have convinced the reader of our slogan `doing more is doing less' in the introduction: rather than solely deforming the commutative algebra $\BB{C}[u,v,w]/\bra f\ket$, we deform the entire derived category of sheaves, which requires less effort, and yet the resulting deformed algebra \eqref{hairy_algebra} is quite involved. However it is unclear to us if the algebra \eqref{hairy_algebra} is of any independent interest.

\section{$T^*\BB{P}^2$}
In the above warmup section, all quiver algebras have global dimension 2 and so there is no potential obstruction to solving the MC equation order by order \eqref{order_by_order}. In this section, we deal with an example of global dimension 4, so as to illustrate the merit of the combinatorial deformation quantisation reviewed in sec.\ref{sec_CDQ}. We first give our motivation.

\subsection{Motivation: Generalised K\"ahler structure}
The generalised K\"ahler structure (GKS) \cite{Gualtieri2014} appeared as the target space of certain 2D supersymmetric sigma model. This geometry mixes complex and symplectic geometry in an organic way and has been expected to play important roles in understanding mirror symmetry.

It was proved in \cite{Gualtieri2014} that GKS is equivalent to the bi-Hermitian geometry $(M,g,I_+,I_-)$ with $I_{\pm}$ being two complex structure whose torsions satisfy some constraints. In the work \cite{Bischoff:2018kzk}, Bischoff, Gualtieri and Zabzine were able to elucidate on the geometric nature of the generalised K\"ahler potential and thereby `solve' the GKS. The key ingredient of their construction involves a holomorphic symplectic groupoid $(X,J,\Go)$, where $M\hookrightarrow X$ serves as the space of units.
From the groupoid one can write down the two complex structures $I_{\pm}$ rather neatly, and more importantly, it was pointed out in \cite{Bischoff:2021ixy} that the embedding $M\hookrightarrow X$ is a much better version of the complexification required in the brane quantisation in \cite{Gukov:2008ve}.

The quantisation via brane picture arises in this setting since one can regard $M$ with its prequantisation line bundle as a brane $B$, while $B_{cc}$ fills the entire $X$. The strings i.e. $\hom(B,B_{cc})$ or $\hom(B_{cc},B_{cc})$ get deformed due to $\Go$. One will also need to deform $\Go\to \Go + F$ for some closed real 2-form (part of the data in solving $I_{\pm}$, see sec. 1 of \cite{Bischoff:2018kzk}. The details of $F$ will not concern us here), which deforms the complex structure of $X$. Both these two deformations can be captured in our current frame work.

Note also $X$ is not just any `complexification' of $M$ as it carries also a symplectic groupoid structure, hence its quantisation also ties into the quantisation via symplectic groupoid approach \cite{Hawkins}. Bischoff and Gualtieri in \cite{Bischoff:2021ixy} implemented this quantisation idea for toric surfaces. In the current work, we cannot reach their level of generality, but for certain GKS on $\BB{P}^2$, we can apply our idea of brane quantisation via minimal resolutions.

In \cite{GKS}, a novel GKS on $\BB{P}^2$ was constructed, whose $I_{\pm}$ are written in terms of elliptic functions. Let us get to the geometric setting that actually concerns us. The holomorphic symplectic groupoid $(X,\Go_0)$ in this story is topologically $T^*\BB{P}^2$, but $\Go_0$ is not the canonical one, it reads in local coordinates
\bea i\Go_0=-\xi_1\xi_2 dx^1\wedge dx^2-\xi_1x^2 dx^1\wedge d\xi_2+dx^1\wedge d\xi_1+dx^2\wedge d\xi_2+\xi_2 x^1 dx^2\wedge d\xi_1-x^1x^2 d\xi_1\wedge d\xi_2\label{hol_symp_P2}.\eea
The notations here are: if $[z^1,z^2,z^3]$ are the homogeneous coordinates, then $x^1=z^1/z^3$, $x^2=z^2/z^3$ are the inhomogeneous coordinates on the open $\{z^3\neq 0\}$. The $\xi_{1,2}$ are the coordinates of the cotangent fibre above $(x^1,x^2)$.
The inverse of $\Go_0$ reads
\bea -i\Pi=\underline{\partial_{\xi_1}\wedge\partial_{x^1}+\partial_{\xi_2}\wedge\partial_{x^2}}-\xi_1\xi_2\partial_{\xi_1}\wedge\partial_{\xi_2}+\xi_1 x^2\partial_{\xi_1}\wedge\partial_{x^2}
-\xi_2 x^1\partial_{\xi_2}\wedge\partial_{x^1}-x^1x^2\partial_{x^1}\wedge\partial_{x^2}\label{Pi_TP2}\eea
The underlined terms correspond to the canonical Poisson structure on a cotangent bundle. The two parts of \eqref{Pi_TP2} can have any relative coefficients and still be Poisson.
We leave it to the reader to show that $\Go_0$ and $\Pi$ are globally defined.

As mentioned earlier, we will need to consider a deformation
\bea \Go_0\to \Go_t=\Go_0+tF.\label{Go+F}\eea
Here $F$ is a closed real 2-form; such non-holomorphic object could have been the kryptonite for our approach that heavily relies on algebraic methods. We will need some workaround to evade this.

To summarise this section and set up the goal: we will perform deformation quantisation on functions of $(X=T^2\BB{P}^2,\Go_0)$ and at the same time take into account the deformation \eqref{Go+F}.

\subsection{The Hochschild Cohomology of $T^*\BB{P}^2$}
We first investigate $\HH^2(X)$ to grasp the infinitesimal deformation of $D^b(X)$.
From the HKR theorem,
\bea \HH^2(X)=H^2(X,{\cal O}_X)\oplus H^1(X,TX)\oplus H^0(X,\wedge^2TX).\nn\eea
Let $V$ be a 3-dim complex vector space and $\BB{P}(V)$ the projective space and $X=T^*\BB{P}(V)$.
The cohomology groups are easy to calculate, some results are collected in the appendix.
In particular, they are graded (by essentially the fibre power). We use $u^2$ as the formal parameter to keep track of this grading.
From prop.\ref{prop_for_HH2}, to the lowest gradings we have
\bea
\HH^2\big|_{u^{-4}}&=&0,\nn\\
\HH^2\big|_{u^{-2}}&=&\BB{C}\oplus\BB{C},\nn\\
\HH^2\big|_{u^0}&=&{\tiny\yng(3)}\oplus {\tiny\yng(3,3)}\oplus 2\,{\tiny\yng(2,1)}.\label{HH2_rep}\eea
Here one of the two summands $\BB{C}$ corresponds to $H^1(X,TX)$ and is the only non-vanishing $H^1$. That $H^1(X,TX)$ be of dimension 1 will be important later when we try to match the deformation \eqref{Go+F}.

Just having the $\HH^2$ is not enough, we will need its explicit cochain representative in the bar complex and in the minimal resolution. Unfortunately, we found no clever way of doing this, and we leave in the appendix a brute force calculation using the minimal resolution.
The upshot is that for each term in \eqref{HH2_rep}, we found an explicit cochain listed in prop.\ref{thm_HH2_0} in the $SL(3)$ representation ${\tiny\yng(3)}$, ${\tiny\yng(3,3)}$ and ${\tiny\yng(2,1)}\times 2$, matching that of \eqref{HH2_rep}.

There remains the problem of finding out how to match the Poisson structure \eqref{Pi_TP2} to $\HH^2$. For this, we try to decompose \eqref{Pi_TP2} into representations of $SL(3)$, and match them to \eqref{HH2_rep}.
\begin{lemma}\label{lem_Pi_rep}
  The first part of $\Pi$ of \eqref{Pi_TP2} transforms in the trivial representation, while the remaining terms transform in ${\tiny\yng(3)}\oplus{\tiny\yng(3,3)}$.
\end{lemma}
\begin{proof}
  Instead of using $x^{1,2}$ and $\xi_{1,2}$ as local coordinates, which obscures the $SL(3)$ symmetry, we use the homogeneous coordinates $[z^1,z^2,z^3;w_1,w_2,w_3]$ with $z^iw_i=0$.
  We claim that $\Pi$ can be written as
  \bea -i\Pi&=&\frac13\sum_{a=1}^3\Big(-z^a\partial_{z^a}\wedge z^{a+1}\partial_{z^{a+1}}+z^a\partial_{z^a}\wedge(w_{a+1}\partial_{w_{a+1}}-w_{a-1}\partial_{w_{a-1}})-w_a\partial_{w_a}\wedge w_{a+1}\partial_{w_{a+1}}\Big)\nn\\
  &&+\partial_{w_a}\wedge \partial_{z^a}.\label{Poisson_hol_pre}\eea
  The first line transforms in either ${\tiny\yng(3)}$ or ${\tiny\yng(3,3)}$. To see this, note the irreducible tensor representations of $SL(3)$ are $T^{a_1\cdots a_m}_{b_1\cdots b_n}$, with the $a,b$ indices symmetric and $T$ is traceless in any pair of $a,b$ indices (see e.g. \cite{Coleman_fun}). Look at $z^a\partial_{z^a}\wedge z^{a+1}\partial_{z^{a+1}}$, it has tensor components $T^{ab}_{cd}$ with $ab$ symmetric,  $cd$ anti-symmetric and $T^{a,a+1}_{a,a+1}=1$. A pair of anti-symmetric contravariant indices can be traded for a covariant one : $T^{ab}_{cd}\ep^{cde}:=D^{abe}$, one can easily see that $D$ is totally symmetric with $D^{123}=1$, so $D\in {\tiny\yng(3)}$. Similarly look at the middle term that corresponds to a tensor structure $T^{ab}_{cd}$ such that $T^{ba}_{dc}=-T^{ab}_{cd}$ and $T^{a,a+1}_{a,a+1}=1$. Thus either $a,b$ are symmetric $c,d$ anti-symmetric or the other way round. One gets ${\tiny\yng(3)}$ or ${\tiny\yng(3,3)}$ in the two cases.

  Finally for the claim, if in \eqref{Poisson_hol_pre} we change coordinates to
  \bea x^{1,2}=z^{1,2}/z^3,~~~\xi_{1,2}=w_{1,2}z^3,~~~\gl=z^3,~~~\mu=z^iw_i,\nn\eea
  then we get
  \bea \eqref{Poisson_hol_pre}&=&-x^1x^2\partial_{x^1}\wedge \partial_{x^2}-\partial_{x^i}\wedge \partial_{\xi_i}+x^1\xi_2\partial_{x^1}\wedge \partial_{\xi_2}-x^2\xi_1\partial_{x^2}\wedge\partial_{\xi^1}-\xi_1\xi_2\partial_{\xi_1}\wedge\partial_{\xi_2}\nn\\
&&+\gl\partial_{\gl}\wedge\partial_{\mu}+\frac13\gl\partial_{\gl}\wedge\big(-x^1\partial_{x^1}+x^2\partial_{x^2}
+\xi_1\partial_{\xi_1}-\xi_2\partial_{\xi_2}\big).\nn\eea
When restricted to functions invariant under the rescaling
\bea [z^1,z^2,z^3;w_1,w_2,w_3]\to [z^1\gl,z^2\gl,z^3\gl;w_1/\gl,w_2/\gl,w_3/\gl]\nn\eea
then \eqref{Poisson_hol_pre} equals $\Pi$.
\end{proof}

What about the deformation \eqref{Go+F}. As $F$ is not holomorphic, we need to reach a compromise to incorporate $F$ and at the same time still be able to use our algebraic approach. But the main purpose of $F$ is to deform the complex structure on $X$. Indeed $\iota_v\Go=0$ for any $v$ that is anti-holomorphic, so we define a new complex structure by declaring $v$ anti-holomorphic if $\iota_v(\Go+tF)=0$. That $\Go,F$ are closed means that the new complex structure is integrable. Also recall that the infinitesimal deformations of complex structures are encoded by $H^1(X,TX)$. In the following, we will treat $F$ as merely serving to deform the complex structure, and thereby match it with $H^1(X,TX)$ and subsequently to $\HH^2(X)$.
\begin{lemma}\label{lem_Kodaira_Spencer}
  The cohomology $H^1(X,TX)$ is 1-dim for $X=T^*\BB{P}(V)$. It corresponds to the deformation from $w_iz^i=0$ to $w_iz^i=t$.
\end{lemma}
\begin{proof}
  That $H^1(X,TX)$ is 1-dim is proved in prop.\ref{prop_for_HH2} where we show that $H^1(\BB{P}(V),T\otimes T^*\BB{P}(V))=\BB{C}$ is the only contribution to $H^1(X,TX)$.

  We also recall that for a hyper-surface, then $H^1(X,TX)$ is realised as a deformation of the equation defining $X$. For our $X$, it can be realised as follows. Consider $(V\backslash\{0\})\times V^*$ with coordinates $z^{1,2,3},w_{1,2,3}$ such that $w_iz^i=0$. One then takes a quotient with respect to a $\BB{C}^*$ action $[\vec z,\vec w]\to [\gl \vec z,\gl^{-1}\vec w]$. Due to the presence of the quotient, we need to go to C\v{e}ch covers and find explicit C\v{e}ch representative of $H^1(X,TX)$.

  We pick the opens $U_i=\{z^i\neq 0\}$ and $U_{ij}=U_i\cap U_j$, then $H^1(X,TX)$ is represented as a vector field locally defined on $U_{ij}$ satisfying a cocycle condition. Explicitly set
  \bea V\big|_{U_{23}}=\frac{1}{x^2}\frac{\partial}{\partial \xi_2}\nn\eea
  where $\xi_{1,2}=w_{1,2}z^3$ and $x^{1,2}=z^{1,2}/z^3$. One makes a cyclically symmetric choice in other $U_{ij}$'s and one can check that this gives the unique representative of $H^1(X,TX)$. These vector fields defined on intersections $U_{ij}$ deforms the transition function $U_i\to U_j$. For example, in $U_2$, pick local coordinates as $y^{1,3}=z^{1,3}/z^2$ and $\eta_{1,3}=w_{1,3}z^2$, with the transition function
  \bea \eta_3=-x^2(x^1\xi_1+x^2\xi_2),~~\eta_1=\xi_1x^2,~~~y^1=x^1/x^2,~~~y^3=1/x^2,~~~V|_{U_{23}}=(1/x^2)\partial_{\xi_2}={\color{blue}-(1/y^3)}\partial_{\eta_3}.\nn\eea
  The deformed transition function reads
  \bea \eta_3=-x^2(x^1\xi_1+x^2\xi_2)+\ep{\color{blue}(-1/y^3)}.\nn\eea
  Rewrite the equation in terms of $z,w$: $w_3z^2=-(z^2/z^3)(z^1w_1+z^2w_2)+\ep(-z^2/z^3)$, we get $\sum w_iz^i=-\ep$.
\end{proof}

\subsection{The Explicit Cocycle Representing $\HH^2$}
In the appendix sec.\ref{sec_ECoHUtMR} we obtain the explicit cochain representing $\HH^2(T^*\BB{P}^2)$
\begin{proposition}\label{thm_HH2_0}
  At order $u^{-2}$, $\HH^2(T^*\BB{P}^2)$ is of dimension 2, whose generators are denoted as $g_{\emptyset,0}$ and $g_{\emptyset,1}$.
  They evaluate on $\SF{T}_2$ in \eqref{T_1234} to give

\bea \begin{array}{|c|c|c|c|c|}
  \hline
         & [z\cdotp w]_{11} & [w\cdotp z]_{11} & [z\cdotp w]_{22} & [w\cdotp z]_{00} \\
         \hline
         g_{\emptyset,0} & e_1 & -e_1 & \frac13e_2 & -\frac13e_0 \\
\hline
g_{\emptyset,1} & e_1 & e_1 & e_2 & e_0 \\
\hline
\end{array},\nn\eea
where $[z\cdotp w]_{11}$ is short for $e_1\otimes z\cdotp w\otimes e_1$.

At order $u^0$, $\HH^2(T^*\BB{P}^2)$ has four components, arranged according to the representation under $SL(3)$: $g_{\tiny\yng(3)}$, $g_{\tiny\yng(3,3)}$, $g_{{\tiny\yng(2,1)},s}$ and $g_{{\tiny\yng(2,1)},a}$. They evaluate on $\SF{T}_2$ to give
\bea &&\begin{array}{|c|c|c|c|c|c|c|}
  \hline
         & [z\cdotp w]_{11} & [w\cdotp z]_{11} & [z\cdotp w]_{22} & [w\cdotp z]_{00} & [z^{[i}z^{j]}]_{20} & [w_{[i}w_{j]}]_{02} \\
         \hline
{\tiny\yng(3)} & 0 & 0 & 0 & 0 & \ep^{ijr}D_{rkl}[z^kz^l]_{20} & 0 \\
\hline
{\tiny\yng(3,3)} & 0 & 0 & 0 & 0 & 0 & -\ep_{ijr}\bar D^{rkl}[w_kw_l]_{02} \\
\hline
{\tiny\yng(2,1)} & U^i_j[w_iz^j]_{11} & -U^i_j[w_iz^j]_{11} & X^i_j[z^jw_i]_{22} & -X^i_j[w_iz^j]_{00} & 0 & 0\\
\hline
\end{array},\nn\\
&&\begin{array}{|c|c|}
\hline
& [w_i,z^j]_{tf,11} \\
\hline
 {\tiny\yng(3)}   &  \ep^{jlr}D_{rik}[w_lz^k]_{11} \\
\hline
 {\tiny\yng(3,3)} &  \ep_{rik}\bar D^{rjl}[w_lz^k]_{11} \\
\hline
{\tiny\yng(2,1)} & (3X-U)^k_i[w_kz^j]_{11}+(3X-U)^j_l[w_iz^l]_{11}-(2X-\frac23U)^k_l\gd^j_i[w_kz^l]_{11} \\
\hline
\end{array}\nn\eea
where $\bar D^{ijk},D_{ijk}$ are totally symmetric and $X^i_j,U^i_j$ are traceless tensors.
\end{proposition}

We have shown in lem.\ref{lem_Pi_rep} that $\Pi$ transforms in the trivial representation and in ${\bf 10}\oplus{\bf 10}^*$. Thus we make a table of the evaluation of the following 2-cocycle
\bea g=ag_{\tiny\yng(3)}+bg_{\tiny\yng(3,3)}+cg_{\emptyset,0}+tg_{\emptyset,1}\label{used_V}\eea
with $\bar D^{123}=1=D_{123}$.
\begin{lemma}
On the tips of the reduction system, $g$ evaluates to
 \bea 1:&&\bra z^1z^2\ket_{20}\to(2a)z^2z^1,~~~\bra z^2z^3\ket_{20}\to(2a)z^3z^2,~~~\bra z^1z^3\ket_{20}\to(-2a)z^3z^1\nn\\
2:&&\bra w_1w_2\ket_{02}\to(-2b)w_2w_1,~~~\bra w_2w_3\ket_{02}\to(-2b)w_3w_2,~~~\bra w_1w_3\ket_{02}\to(2b)w_3w_1\nn\\
3:&&\bra z^1w_2\ket_{11}\to(b-a)w_2z^1,~~~\bra z^1w_3\ket_{11}\to(a-b)w_3z^1,~~~\bra z^2w_3\ket_{11}\to(b-a)w_3z^2\nn\\
4:&&\bra z^2w_1\ket_{11}\to(a-b)w_1z^2,~~~\bra z^3w_1\ket_{11}\to(b-a)w_1z^3,~~~\bra z^3w_2\ket_{11}\to(a-b)w_2z^3\nn\\
5:&&\bra z^1w_1\ket_{11}\to(a+b)w_3z^3+(-a-b)w_2z^2+\frac{2c}{3}\nn\\
6:&&\bra z^2w_2\ket_{11}\to(a+b)w_1z^1+(-a-b)w_3z^3+\frac{2c}{3}\nn\\
7:&&\bra z^3w_3\ket_{11}\to(-a-b)w_1z^1+(a+b)w_2z^2+(t-\frac{c}{3})\nn\\
8:&&\bra w_3z^3\ket_{11}\to(t-c),~~~\bra w_3z^3\ket_{00}\to(t-\frac{c}{3}),~~~\bra z^3w_3\ket_{22}\to(\frac{c}{3}+t).\label{2-cocycle}\eea
\end{lemma}
We stress that the lhs of every item is a tip of a reduction system, while the rhs lives in $\Gl$ (i.e. all relations in $\Gl$ are enforced). Also we did not write the idempotents on the rhs as they concur with the lhs.
\begin{proof}
  The proof is just an application of the table in \eqref{thm_HH2_0}. We will give some examples of how \eqref{2-cocycle} is obtained, as there is some subtlety.
  For the tip $\bra z^1z^2\ket _{20}$, recall that it represents the relation $[z^{[1}z^{2]}]_{20}=e_2\otimes z^{[1}z^{2]}\otimes e_0$ in \eqref{T_1234}, looking at the table, it evaluates to
  \bea [z^{[i}z^{j]}]_{20}\stackrel{g}{\to} \ep^{ijr}C_{rkl}[z^kz^l]_{20}=2\ep^{123}D_{312}[z^1z^2]_{20}=2[z^2z^1]_{20}.\nn\eea
  Take another example $[z^1w_1]_{11}$, it represents the relation $e_1\otimes [z^1,w_1]\otimes e_1=e_1\otimes [z^1,w_1]_{tf}\otimes e_1+(1/3)e_1\otimes (z^iw_i-w_iz^i)\otimes e_1$, where we have separated the trace free and trace part. It evaluates to
  \bea [z^1w_1-w_1z^1]_{11}&\stackrel{g}{\to}&-a\ep^{1lr}D_{r1k}[w_lz^k]_{11}-b\ep_{r1k}\bar D^{r1l}[w_lz^k]_{11}+\frac13(c+t)-\frac13(-c+t)\nn\\
  &=&-a[w_2z^2-w_3z^3]_{11}-b[-w_3z^3+w_2z^2]_{11}+\frac23c.\nn\eea
\end{proof}
The 2-cocyle evaluated on the remaining superfluous tips listed in Eqs.\ref{rel_wzw}, \ref{rel_zwz} and \ref{new_rel} are completely fixed by the list above, but we record these for the convenience of the reader
\bea
9:&&\bra z^1w_1z^2\ket_{21}\to(-2a-2b)z^2w_2z^2+(\frac{2c}{3})z^2\nn\\
10:&&\bra z^1w_2z^2\ket_{21}\to(-2a-2b)z^1w_1z^1+(-\frac{2c}{3})z^1\nn\\
11:&&\bra z^1w_3z^2\ket_{21}\to(4a-2b)z^2w_3z^1\nn\\
12:&&\bra z^1w_1z^3\ket_{21}\to(-2a-2b)z^3w_2z^2+(\frac{2c}{3})z^3+(-2a-2b)z^3w_1z^1\nn\\
13:&&\bra z^1w_2z^3\ket_{21}\to(2b-4a)z^3w_2z^1\nn\\
14:&&\bra z^2w_1z^3\ket_{21}\to(4a-2b)z^3w_1z^2\nn\\
15:&&\bra z^2w_2z^3\ket_{21}\to(2a+2b)z^3w_1z^1+(\frac{2c}{3})z^3+(2a+2b)z^3w_2z^2\nn\eea
\bea
16:&&\bra w_1z^1w_2\ket_{12}\to(2a+2b)w_2z^2w_2+(-\frac{2c}{3})w_2\nn\\
17:&&\bra w_1z^2w_2\ket_{12}\to(2a+2b)w_1z^1w_1+(\frac{2c}{3})w_1\nn\\
18:&&\bra w_1z^3w_2\ket_{12}\to(2a-4b)w_2z^3w_1\nn\\
19:&&\bra w_1z^1w_3\ket_{12}\to(-\frac{2c}{3})w_3+(2a+2b)w_3z^1w_1+(2a+2b)w_3z^2w_2\nn\\
20:&&\bra w_1z^2w_3\ket_{12}\to(4b-2a)w_3z^2w_1\nn\\
21:&&\bra w_2z^1w_3\ket_{12}\to(2a-4b)w_3z^1w_2\nn\\
22:&&\bra w_2z^2w_3\ket_{12}\to(-\frac{2c}{3})w_3+(-2a-2b)w_3z^2w_2+(-2a-2b)w_3z^1w_1\nn\\
23:&&\bra w_3w_1z^3\ket_{01}\to(t-c)w_1+(4b)w_2w_1z^2+(2b)w_1w_1z^1\nn\\
24:&&\bra w_3w_2z^3\ket_{01}\to(t-c)w_2+(-2b)w_2w_1z^1+(-2b)w_2w_2z^2\label{2-cocycle_more}.\eea

\begin{lemma}
  The Poisson structure (recall that $\Pi_s$ is Poisson for any $s$)
  \bea -i\Pi_s=s(\partial_{\xi_1}\wedge\partial_{x^1}+\partial_{\xi_2}\wedge\partial_{x^2})-\xi_1\xi_2\partial_{\xi_1}\wedge\partial_{\xi_2}+\xi_1 x^2\partial_{\xi_1}\wedge\partial_{x^2}
 -\xi_2 x^1\partial_{\xi_2}\wedge\partial_{x^1}-x^1x^2\partial_{x^1}\wedge\partial_{x^2}\nn\eea
  correspond to the parameters
\bea a=-\frac16,~~b=\frac16,~~~c=-\frac{3s}{2}.\nn\eea
Finally, the parameter $t$ corresponds to the Kodaira-Spencer deformation $w_iz^i=0$ $\to$ $w_iz^i=t$.
\end{lemma}
\begin{proof}
First the statement about $t$ and $g_{\emptyset,1}$ is quite obvious.

Next we consider paths starting and ending at $e_0$, which are functions on $X$. Then we match the Poisson bracket e.g. $\{w_1z^3,w_3z^1\}$ with the commutator $[w_1z^3|w_3z^1]_{00}-[w_3z^1|w_1z^3]_{00}$ evaluated on \eqref{used_V}.
First $w_1z^3=\xi_1$ and $w_3z^1=-\xi_1(x^1)^2-\xi_2x^1x^2$, hence
\bea \{w_1z^3,w_3z^1\}=-s(2\xi_1x^1+\xi_2x^2)=-s(2w_1z^1 + w_2z^2).\label{used_VI}\eea
On the other hand $[w_1z^3|w_3z^1]_{00}$ evaluates to
\bea [w_1z^3|w_3z^1]_{00}&\stackrel{G_2}{\to}&w_1\otimes z^3w_3\otimes z^1-e_0\otimes w_1w_2\otimes z^2z^1\nn\\
&\stackrel{g}{\to}&(-a-b)w_1w_1z^1z^1+(a+3b)w_2w_1z^2z^1+(t-\frac{c}{3})w_1z^1,\nn\\
\left[w_3z^1|w_1z^3\right]_{00}&\stackrel{G_2}{\to}&w_3\otimes z^1w_1\otimes z^3+w_3w_1\otimes z^1z^3\otimes e_0+e_0\otimes w_3w_1z^3\otimes z^1\nn\\
&\stackrel{g}{\to}&(t - \frac{5c}{3})w_1z^1-\frac{2c}{3}w_2z^2+3(a+b)w_1w_1z^1z^1+(5a+7b)w_2w_1z^2z^1+2(a+b)w_2w_2z^2z^2\nn\eea
thus we get
\bea
&&gG_2(\left[w_1z^3|w_3z^1\right]_{00}-\left[w_3z^1|w_1z^3\right]_{00})\nn\\
&&\hspace{2cm}=\frac{4c}{3}w_1z^1+\frac{2c}{3}w_2z^2-2(a+b)w_{2}w_2z^2z^2-4(a+b)w_{1}w_1z^1z^1-4(a+b)w_2w_1z^2z^1.\nn\eea
For this to match \eqref{used_VI}, we need $a=-b$, $c=-3s/2$. To further fix $a,b$, we consider
\bea \left[w_1z^3|w_2z_1\right]_{00}&\stackrel{G_2}{\to}&w_1\otimes z^3w_2\otimes z^1+e_0\otimes w_1w_2\otimes z^3z^1\nn\\
&\stackrel{g}{\to}&(a-3b)w_2w_1z^3z^1,\nn\\
\left[w_2z_1|w_1z^3\right]_{00}&\stackrel{G_2}{\to}&w_2\otimes z^1w_1\otimes z^3+w_2w_1\otimes z^1z^3\otimes e_0\nn\\
&\stackrel{g}{\to}&\frac{2c}{3}w_2z^3-(3a + b)w_2w_1z^3z^1-2(a + b)w_2w_2z^3z^2\nn\eea
Thus if
\bea gG_2(\left[w_1z^3|w_2z_1\right]_{00}-\left[w_2z_1|w_1z^3\right]_{00})=(4a-2b)w_2w_1z^3z^1-\frac{2c}{3}w_2z^3+2(a + b)w_2w_2z^3z^2\nn\eea
is to match $\{w_1z^3,w_2z_1\}=\{\xi_1,\xi_2x^1\}=s\xi_2-\xi_1\xi_2x^1=sw_2z^3-w_2w_1z^3z^1$, we need $b=-a=1/6$.

We have checked that the same $a,b,c$ work for all length two paths.
We point out that the reason why such a match is possible follows from the HKR map \eqref{HKR}, plus the matching of $SL(3)$ representations \eqref{HH2_rep} and lem.\ref{lem_Pi_rep}. But the extra checks do confer a reassuring nod to our result.
\end{proof}
From these explicit cocycles, one can make a pitch for the combinatorial star product. In reality, we solved the MC-equation to the fifth order in $\hbar$ to make sure nothing untoward lurks in the codes
\begin{theorem}
 We have the following all order multiplication rule (where $q=\exp(\hbar/3)$ and we omit the $\star$)
\bea 1:&&(z^1z^2)_{20}= q^{-1}z^2z^1,~~~(z^2z^3)_{20}= q^{-1}z^3z^2,~~~(z^1z^3)_{20}= qz^3z^1\nn\\
2:&&(w_1w_2)_{02}= q^{-1}w_2w_1,~~~(w_2w_3)_{02}= q^{-1}w_3w_2,~~~(w_1w_3)_{02}= qw_3w_1\nn\\
3:&&(z^1w_2)_{11}= qw_2z^1,~~~(z^1w_3)_{11}= q^{-1}w_3z^1,~~~(z^2w_3)_{11} = qw_3z^2\nn\\
4:&&(z^2w_1)_{11}= q^{-1}w_1z^2,~~~(z^3w_1)_{11}= qw_1z^3,~~~(z^3w_2)_{11}= q^{-1}w_2z^3\nn\\
5:&&(z^1w_1)_{11}= w_1z^1-s,~~~(z^2w_2)_{11}= w_2z^2-s\nn\\
6:&&(z^3w_3)_{11}=-w_1z^1-w_2z^2+t+\frac{s}{2}\nn\\
7:&&(w_3z^3)_{11}=-w_1z^1-w_2z^2+(t+\frac{3s}{2})\nn\\
8:&&(w_3z^3)_{00}=-w_1z^1-w_2z^2+(t+\frac{s}{2})\nn\\
9:&&(z^3w_3)_{22}=-z^1w_1-z^2w_2+(t-\frac{s}{2})\label{mult_rule}.\eea
\end{theorem}
We have been skimping on notations above, if spelled out fully, the rule should read e.g. $e_2z^1\star z^2e_0=q^{-1}e_2z^2\star z^1e_0$, etc.

For the remaining superfluous tips, their reductions are determined from \eqref{mult_rule}, but we list them still
\bea 10:&&(z^1w_1z^2)_{21}= z^2w_1z^1-sz^2,~~~(z^1w_2z^2)_{21}= z^2w_2z^1+sz^1\nn\\
11:&&(z^1w_3z^2)_{21}= q^{-3}z^2w_3z^1,~~~(z^1w_1z^3)_{21}= z^3w_1z^1-sz^3\nn\\
12:&&(z^1w_2z^3)_{21}= q^3z^3w_2z^1,~~~(z^2w_1z^3)_{21}= q^{-3}z^3w_1z^2\nn\\
13:&&(z^2w_2z^3)_{21} = -sz^3,~~~(w_1z^1w_2)_{12}= sw_2\nn\\
14:&&(w_1z^2w_2)_{12}= -sw_1,~~~(w_1z^3w_2)_{12}= q^{-3}w_2z^3w_1\nn\\
15:&&(w_1z^1w_3)_{12}= sw_3,~~~(w_1z^2w_3)_{12}= q^3w_3z^2w_1\nn\\
16:&&(w_2z^1w_3)_{12}= q^{-3}w_3z^1w_2,~~~(w_2z^2w_3)_{12}=sw_3\nn\\
17:&&(w_3w_1z^3)_{01}= q^{-1}(t+\frac{3s}{2})w_1-q^{-2}w_2w_1z^2-q^{-1}w_1w_1z^1\nn\\
18:&&(w_3w_2z^3)_{01}= q(t+\frac{3s}{2})w_2-qw_2w_1z^1-qw_2w_2z^2\label{mult_rule_more}.\eea
\begin{proof}
  The proof of this theorem consists of checking the associativity on the set of ambiguities \eqref{Amb}. We just give one example $\bra w_3z^3w_1\ket_{01}$.
\bea &&(w_3z^3)w_1=(-w_1z^1-w_2z^2+t+\frac{s}{2})w_1=-w_1(w_1z^1-s)-q^{-1}w_2w_1z^2+(t+\frac{s}{2})w_1\nn\\
&&w_3(z^3w_1)=qw_3w_1z^3=(t+\frac{3s}{2})w_1-q^{-1}w_2w_1z^2-w_1w_1z^1.\nn\eea
We get a match. Or if we take $\bra z^1z^3w_3\ket_{21}$,
\bea &&(z^1z^3)w_3=qz^3z^1w_3=z^3w_3z^1=(-z^1w_1-z^2w_2+(t-\frac{s}{2}))z^1,\nn\\
&&z^1(z^3w_3)=z^1(-w_1z^1-w_2z^2+t+\frac{s}{2})=-z^1w_1z^1-(z^2w_2z^1+sz^1)+tz^1+\frac{s}{2}z^1.\nn\eea

But let us recapitulate the line of reasoning that led to this simple proof. For deformation quantisation, one needs to solve the MC-equation \eqref{MC} in the bar complex $\B^{\sbullet}$. The smaller Chouhy-Solotar model $\P^{\sbullet}$ has maps to $\B^{\sbullet}$ given by a tree formula \eqref{tree_formula}. Reading a tree from the root up gives us a multiplication rule $\star$ as in \eqref{star}, which consists of stacking two words together and use the reduction system to reduce the concatenation to an irreducible word.
Finally thm.\ref{thm_ass_chk} says that checking the MC-equation amounts to checking the associativity of $\star$ on the set of ambiguities.
\end{proof}
\begin{remark}
When solving the MC-equation at each order in $\hbar$, there are usually some freedoms. Some of these can be `gauged away' by using a gauge transformation (see \cite{Getzler_Lie}). Others are simply the freedom to redefine the parameters at each order, which rather resembles the `renormalisation of mass and couplings' in physics. But we have not systematically analysed if there is a moduli space of gauge inequivalent solutions to the MC-equation, i.e. a moduli space of $\star$ product.
\end{remark}
\begin{corollary}
  From the deformed quiver algebra $\Gl_{q,s,t}$, we obtain a deformation of the holomorphic functions on $X=T^*\BB{P}^2$, listed in \eqref{deformed_TCP2}.
\end{corollary}
\begin{proof}
  We need only consider $e_0\Gl e_0$, which are precisely $H^0(X)$, the holomorphic functions on $X$. We point out that $H^0(X)$ has the following decomposition as $SL(3)$ representations
  \bea H^0(X)=\emptyset\oplus {\tiny\yng(2,1)}\oplus {\tiny\yng(4,2)} \oplus {\tiny\yng(6,3)}\oplus\cdots\nn\eea
  and is generated as the symmetric tensor product of ${\tiny\yng(2,1)}$.

  We name the basis of ${\tiny\yng(2,1)}\sim (V\otimes V^*)_{tf}$ as $X_i^j$ with $\Tr[X]=0$. In fact, $X_i^j\simeq e_0 w_iz^je_0$. We can easily write down the deformed multiplication rule between the $X$'s from \eqref{mult_rule}.
\end{proof}
To obtain some modules over $\Gl_{q,s,t}$, we consider again $e_0\Gl_{q,s,t}e_0$ and impose
\bea &&e_0X_i^j=0,~~i>j,\nn\\
&&e_0X_2^3=e_0X_2^2=e_0X_3^3=0,\nn\\
&&e_0X_1^1=-ns,~~~e_0(X_1^2)^{n+1}=0\nn\eea
to mimic a highest weight module over $U_{\FR{sl}(3)}$.
This produces a module only when the integrality condition below is satisfied
\bea t=-(n+\frac12)s.\nn\eea
This module corresponds to the representation $S^nV^*$ of $SL(3)$ and hence gives the brane supported on the zero section $\BB{P}(V)\subset T^*\BB{P}(V)$, with the pre-quantisation line bundle ${\cal O}(n)$ attached.

The deformed algebra \eqref{deformed_TCP2} bears certain resemblance to the quantum group $U_q\FR{sl}(3)$, though also quite different. For example, we can identify
\bea e_1=-\frac{1}{s}X_1^2,~~f_1=-\frac{1}{s}X_2^1;~~e_2=-\frac{1}{s}X_2^3,~~f_2=-\frac{1}{s}X_3^2;~~e_3=-\frac{1}{s}X_1^3,~~f_3=-\frac{1}{s}X_3^1.\nn\eea
Then $e_3$ (and $f_3$) are given by the $q$-commutator
\bea e_1e_2-q^{-3}e_2e_1=e_3,~~~e_1e_3-q^{-3}e_3e_1=0=e_2e_3-q^3e_3e_2.\nn\eea
Compare with e.g. Eq.9 in \cite{Burroughs:1989dj}. But the remaining Serre relations such as $[e_1,f_1]$ cannot be reproduced.

\smallskip

As a closing remark, in this work all quiver algebras are homogeneous, it would seem that this ensures that the deformation quantisation give a convergent power series in $\hbar$. One can hope to apply the same technique to the quantisation of Higgs moduli space, as in \cite{Gaiotto:2021kma} by Gaiotto and Witten, who investigated the analytical version of the Langland duality. Perhaps more interestingly, if in some scenario one obtains a divergent power series, then one can apply the resurgence to do a re-sum. As the quantisation is supposed to be computable from the A-model on $(X,\im\Go)$. The resurgence may capture some interesting instanton effects in the $A$-model \cite{Gu:2021ize}.

\appendix

\section{Explicit Cochain Representative of $\HH^2(T^*\BB{P}^2)$}
The Hochschild cohomology groups themselves are not hard to compute, but for deformation quantisation, we will need their explicit cochain representative in the bar and the minimal resolutions. But we start with the easier task of computing some sheaf cohomologies.

\subsection{Collection of Some Sheaf Cohomology Groups on $T^*\BB{P}^2$}
Write $\BB{P}^2=\BB{P}(V)$ for some 3-dim vector space $V$ and $X=T^*\BB{P}(V)$.
Our goal is to compute
\bea H^1(X,TX),~~~H^0(X,\wedge^2TX).\nn\eea
We clearly have
\bea &&TX\simeq \pi^*(T\BB{P}(V)\oplus T^*\BB{P}(V)),\nn\\
&&\wedge^2TX\simeq \pi^*(\wedge^2T\BB{P}(V)\oplus \wedge^2T^*\BB{P}(V)\oplus (T\otimes T^*)\BB{P}(V)),\nn\eea
where $\pi:\,T^*\BB{P}(V)\to \BB{P}(V)$ is the projection. We will omit $\pi^*$ next.
The main tool for the computation is the projection formula
\bea H^{\sbullet}(T^*\BB{P}(V),\pi^*F)=\oplus_{n\geq0}H^{\sbullet}(\BB{P}(V),S^nT\BB{P}(V)\otimes F)\label{projection}\eea
where $F$ is a sheaf on the base.
This way we reduce all calculations from $X$ to $\BB{P}(V)$, for which the two Euler sequences are the main tools
\bea 0\to{\cal O}(-1)\to V\otimes{\cal O}\to T(-1)\to0\label{Euler_seq_dual}\\
0\to\Go(1)\to V^*\otimes{\cal O}\to{\cal O}(1)\to0.\label{Euler_seq}\eea
We have some classical results e.g. $H^0(\BB{P}(V),{\cal O}(n))=S^nV^*$, $H^{>0}(\BB{P}(V),{\cal O}(n))=0$ for $n\geq0$ and  $H^{\sbullet}(\BB{P}(V),{\cal O}(-1))=H^{\sbullet}(\BB{P}(V),{\cal O}(-2))=0$.
Combining these with the Euler sequences, we can compute
\begin{proposition}\label{prop_some_coh}
  \begin{enumerate}
  \item $H^{\sbullet}(\BB{P}(V),T^*(1))=0$,
  \item $H^{\sbullet}(\BB{P}(V),T)=(V\otimes V^*)_{tf}[0]$,
  \item $H^{\sbullet}(\BB{P}(V),T^*)=\BB{C}[1]$,
  \item $H^{\sbullet}(\BB{P}(V),T\otimes T^*)=\BB{C}[0],$
  \item $H^{\sbullet}(\BB{P}(V),T^*(2))=\wedge^2V^*[0]$,
  \item $H^{\sbullet}(\BB{P}(V),T^*(3))=\ker(V^*\otimes S^2V^*\to S^3V^*)[0]$,
  \item $H^{>0}(\BB{P}(V),T^*(n))=0$, $n\geq 1$.
\end{enumerate}
\end{proposition}
We can use the Young tableau notation
\bea V^*\simeq {\tiny\yng(1)},~~~S^3V^*={\tiny\yng(3)},~~~(V\otimes V^*)_{tf}\simeq {\tiny\yng(2,1)},~~~\ker(V^*\otimes S^2V^*\to S^3V^*)\simeq {\tiny\yng(2,1)}.\nn\eea
\begin{proof}
  First the long exact sequence (LES) associated with \eqref{Euler_seq} gives
  \bea 0\to H^0(\BB{P}(V),T^*(1))\to V^* \stackrel{\simeq}{\to} V^*\to H^1(\BB{P}(V),T^*(1))\to0 \to \cdots\nn\eea
  showing that $H^{\sbullet}(\BB{P}(V),T^*(1))=0$.

  Now tensor ${\cal O}(1)$ to \eqref{Euler_seq_dual}, the resulting LES gives
  \bea 0\to\BB{C}\to V\otimes V^* \to H^0(\BB{P}(V),T)\to 0\nn\eea
  showing that $H^0(\BB{P}(V),T)$ is the trace-free part of $V\otimes V^*$ and the higher cohomologies vanish. Similarly, tensoring ${\cal O}(-1)$ to \eqref{Euler_seq} gives the third result.

  Tensor $T^*(1)=\Go(1)$ to \eqref{Euler_seq_dual}: $0\to T^*\to V\otimes T^*(1)\to T\otimes T^*\to0$. Since $H^{\sbullet}(\BB{P}(V),T^*(1))=0$, we get
  \bea H^i(T\otimes T^*)\simeq H^{i+1}(T^*).\nn\eea
  From item 3, we get the fourth result.

  For item 5, 6 we tensor ${\cal O}(1),{\cal O}(2)$ to \eqref{Euler_seq}, the rest is similar. Finally item 7 also follows from \eqref{Euler_seq}.
\end{proof}

In addition to the calculations above, one still needs to tensor with $S^nT$ as in \eqref{projection}.
The same calculation with Euler sequence becomes cumbersome. But when $n$ gets large enough, the Kodaira vanishing enters and one can just compute the cohomology via the index theorem. The calculation above merely serves to keep track of the higher cohomologies happening at lower $n$.

\begin{proposition}\label{prop_for_HH2}
We have the following contribution to $\HH^2(T^*\BB{P}(V))$ at order $u^{-2}$ and $u^0$
  \bea 
  -2:&&H^0(\BB{P}(V),T^*\otimes T)=\BB{C},~~H^1(\BB{P}(V),T^*)=\BB{C},\nn\\
   0:&&H^0(\BB{P}(V),T\otimes(T^*\otimes T))=2{\tiny\yng(2,1)}\oplus {\tiny\yng(3,3)},~~H^0(\BB{P}(V),\wedge^2T)={\tiny\yng(3)}\nn\eea
   Importantly $H^{>0}(\BB{P}(V),S^nT\otimes T^*)=0$ for all $n>0$.
\end{proposition}
Here the grading correspond to the fibre power and eventually will correspond to the path length in the quiver associated to $T^*\BB{P}(V)$.
\begin{proof}
   We compute
   \bea H^0(\BB{P}(V),T\otimes(T^*\otimes T))=H^0(\BB{P}(V),\wedge^2T\otimes T^*)\oplus H^0(\BB{P}(V),S^2T\otimes T^*)\nn\eea
   The first term equals $H^0(\BB{P}(V),T^*(3))={\tiny\yng(2,1)}$ as we computed already. For the second term, we tensor \eqref{Euler_seq_dual} to itself and get
   \bea 0\to V(-1)\to S^2V\to S^2T(-2)\to0~~\To~~0\to V(1)\to S^2V(2)\to S^2T\to0.\nn\eea
   Take the tensor with $T^*$, compute the LES
   \bea \to V\otimes H^{\sbullet}(\BB{P}(V),T^*(1))\to S^2V\otimes H^{\sbullet}(\BB{P}(V),T^*(2))\to H^{\sbullet}(\BB{P}(V),S^2T\otimes T^*)\to \nn\eea
   Using $H^{\sbullet}(\BB{P}(V),T^*(1))=0$, we have
   \bea H^0(\BB{P}(V),S^2T\otimes T^*)\simeq S^2V\otimes H^0(\BB{P}(V),T^*(2))={\tiny\yng(2,2)}\otimes{\tiny\yng(1,1)}={\tiny\yng(2,1)}\oplus {\tiny\yng(3,3)}\nn\eea
   and further $H^{>0}(\BB{P}(V),S^2T\otimes T^*)=0$.

   Finally to establish $H^{>0}(\BB{P}(V),S^nT\otimes T^*)=0$ for all $n>1$ ($n=2$ is already shown), we can argue with Kodaira vanishing or directly. Consider $0\to S^{n-1}V\otimes{\cal O}(-1)\to S^nV\otimes{\cal O}\to S^nT(-n)\to 0$, tensor it with $T^*(n)$: $0\to S^{n-1}V\otimes T^*(n-1)\to S^nV\otimes T^*(n)\to S^nT\otimes T^*\to 0$. The left and middle term have no $H^{>0}$ for $n\geq2$, and so $H^{>0}(\BB{P}(V),S^nT\otimes T^*)=0$.
\end{proof}

\subsection{Explicit Calculation of $\HH^2$ Using the Minimal Resolution}\label{sec_ECoHUtMR}
This section contains a tour de force proof of prop.\ref{thm_HH2_0}.

Let us compute first $\partial\SF{T}_3$,
\bea
\partial (\SF{T}_3^1)^j&=&\partial e_2\otimes(z\cdotp w z^j-z^kz^jw_k+z^jz^kw_k)\otimes e_1\nn\\
&=&e_2z^k\otimes[w_k,z^j]\otimes e_1+e_2z^j\otimes(z\cdotp w)\otimes e_1
-e_2\otimes (z\cdotp w)\otimes z^je_1-e_2\otimes[z^j,z^k]\otimes w_ke_1\nn\\
&=&e_2z^k\otimes[w_k,z^j]_{t.f.}\otimes e_1+\frac13e_2z^j\otimes w\cdotp z\otimes e_1+\frac23e_2z^j\otimes(z\cdotp w)\otimes e_1\nn\\
&&-e_2\otimes (z\cdotp w)\otimes z^je_1-e_2\otimes[z^j,z^k]\otimes w_ke_1.\label{used_IV}\eea
Evaluate the rhs on $g_{\emptyset,0}$ and $g_{\emptyset,1}$
\bea
g_{\emptyset,0}(\partial (\SF{T}_3^1)^j)=(-\frac13+\frac23-\frac13)e_2z^je_1=0,\nn\\
g_{\emptyset,1}(\partial (\SF{T}_3^1)^j)=(\frac13+\frac23-1)e_2z^je_1=0.\nn\eea
One can likewise check the closure for other terms in $\SF{T}_3$. Observe that $g_{\emptyset,0},g_{\emptyset,1}$ cannot be coboundaries, since the latter cannot reduce the path length by 2.

Let us now consider cocyles of order $u^0$, i.e. those that do not change path lengths. Thus e.g. $[z\cdotp w]_{11}$ must evaluate to a length 2 path $e_1\ot\ot e_1$. The two arrows must consist of one $z^i\in V^*\simeq{\tiny\yng(1)}$ and one $w_i\in V\simeq{\tiny\yng(1,1)}$. The traceless
condition $z^iw_i$ in $\Gl$ says that $e_1\ot\ot e_1\in {\tiny\yng(2,1)}$. We conclude
\bea \hom_{\Gl^e}([z\cdotp w]_{11},\Gl)\big|_{u^0}=U^i_j[w_iz^j]_{11}\in{\tiny\yng(2,1)}\nn\eea
where $U$ is a traceless tensor. The analysis on $[w\cdotp z]_{11},[w\cdotp z]_{00},[z\cdotp w]_{22}$ is entirely the same.

Similarly for $[z^{[i}z^{j]}]_{20} \in {\tiny\yng(1,1)}$, it must be mapped to $e_2\stackrel{z^k}{\ot}\stackrel{z^l}{\ot}e_0$ with $k,l$ symmetric, thus
\bea \hom_{\Gl^e}([z^{[i}z^{j]}]_{20},\Gl)\big|_{u^0}\in {\tiny\yng(1)}\otimes {\tiny\yng(2)}\simeq {\tiny\yng(3)}\oplus {\tiny\yng(2,1)},\label{used_II}\\
\hom_{\Gl^e}([w_{[i}w_{j]}]_{02},\Gl)\big|_{u^0}\in {\tiny\yng(1,1)}\otimes {\tiny\yng(2,2)}\simeq {\tiny\yng(3,3)}\oplus {\tiny\yng(2,1)}.\nn\eea
Finally for
$([w_i,z^j]_{tf})_{11}\in {\tiny\yng(2,1)}$. It must likewise evaluate to $e_1\ot\ot e_1\in {\tiny\yng(2,1)}$, thus
\bea \hom_{\Gl^e}(([w_i,z^j]_{tf})_{11},\Gl)\big|_{u^0}\in{\tiny\yng(2,1)}\otimes{\tiny\yng(2,1)}\simeq
\tiny{\yng(4,2)}_{27,s}\oplus \tiny{\yng(3)}_{10,a}\oplus \tiny{\yng(3,3)}_{10^*,a}\oplus 2\tiny{\yng(2,1)}_{8,s,a}\oplus \emptyset_{1,s}.\label{used_III}\eea
where the foot number is the dimension and $a,s$ signifies whether the rhs is (anti-)symmetric when exchanging the two copies ${\tiny\yng(2,1)}$ in the tensor product.

We consider now $\SF{T}_3$ and the maps $\partial:\,\SF{T}_3\to \SF{T}_2$. Take $\SF{T}_3^1$, a 0-grading cochain must evaluate it to a length 3 path $e_2\ot\ot\ot e_1$. The three arrows consist of two $z$'s and one $w$. The two $z$'s transform in $S^2V^*\simeq {\tiny\yng(2)}$ (due to commutativity between the $z$'s), while $w$ is in $V\simeq {\tiny\yng(1,1)}$. Thus $e_2\ot\ot\ot e_1\in {\tiny\yng(3,1)}\oplus {\tiny\yng(1)}$. But the last component vanishes due to the tracelessness of $z^iw_i=0$. As $\SF{T}_3^1\in {\tiny\yng(1)}$, we conclude that
\bea \hom_{\Gl^e}(\SF{T}_3^1,\Gl)\big|_{u^0}\in {\tiny\yng(1,1)}\otimes {\tiny\yng(3,1)}={\tiny\yng(2,1)}\oplus {\tiny\yng(3)}\oplus {\tiny\yng(4,2)}.\nn\eea
Next we consider the image of
\bea \hom_{\Gl^e}(\SF{T}_2,\Gl)\big|_{u^0}\to \hom_{\Gl^e}(\SF{T}_3,\Gl)\big|_{u^0}\label{delta_32}\eea
induced from $\partial$.
Since $\partial$ is a map of $SL(3)$ representations, we can analyse the above map for each representation individually.

We start with ${\bf 27}=\tiny{\yng(4,2)}$, as it appears only in \eqref{used_III}. Explicitly it maps
\bea \left[w_i,z^j\right]_{tf,11} \to Z^{jl}_{ik}[w_lz^k]_{11},\nn\eea
where $Z^{jl}_{ik}$ is symmetric in $jl$ and $ik$ and traceless (as befits the tensor representation ${\bf 27}$).
Under the map \eqref{delta_32}, we get using the explicit formula \eqref{used_IV} for $\partial(\SF{T}_3^1)^j$
\bea \partial(\SF{T}_3^1)^j\to Z^{jl}_{ik}[z^iw_lz^k]_{21}\neq 0.\nn\eea
We conclude that the ${\bf 27}$ component must be absent.

We deal with the singlet in \eqref{used_III}. Explicitly, it gives the evaluation
\bea  \left[w_i,z^j\right]_{tf,11} \to [w_iz^j]_{11}.\nn\eea
Under the map \eqref{delta_32}, we get
\bea \partial (\SF{T}_3^1)^j\to [z^kw_kz^j]_{21}=0.\nn\eea
So the singlet is a cocycle, but it will turn out to be a coboundary later.

We deal with ${\bf 10}={\tiny\yng(3)}$ or ${\bf 10}^*={\tiny\yng(3,3)}$, which we recall are symmetric rank 3 covariant or contravariant tensors $D_{ijk},\overline{D}^{ijk}$.

We name the ${\bf 10}$ component in \eqref{used_II} and \eqref{used_III} as $C_{ijk}$ and $D_{ijk}$, explicitly this means that
\bea [z^{[i}z^{j]}]_{20}\to \ep^{ijr}C_{rkl}[z^kz^l]_{20},\nn\\
\left[w_i,z^j\right]_{tf,11} \to \ep^{jlr}D_{rik}[w_lz^k]_{11}.\nn\eea
Under the map \eqref{delta_32}, we get
\bea
\partial (\SF{T}_3^1)^j&=&e_2z^k\otimes[w_k,z^j]_{t.f.}\otimes e_1+\frac13e_2z^j\otimes w\cdotp z\otimes e_1+\frac23e_2z^j\otimes(z\cdotp w)\otimes e_1\nn\\
&&-e_2\otimes (z\cdotp w)\otimes z^je_1-e_2\otimes[z^j,z^k]\otimes w_ke_1,\nn\\
&\to&\ep^{jlr}D_{rik}[z^iw_lz^k]_{21}+0+0+0+\ep^{ijr}C_{rkl}[z^kz^lw_i]_{21}.\nn\eea
For this to be zero, we need $C=D$. The analysis for ${\bf 10}^*$ is similar: we take
\bea [w_{[i}w_{j]}]_{02}\to \ep_{ijr}\bar C^{rkl}[w_kw_l]_{02},\nn\\
\left[w_i,z^j\right]_{tf,11} \to \ep_{rik}\bar D^{rjl}[w_lz^k]_{11}.\nn\eea
Under the map \eqref{delta_32}, we get
\bea \partial(\SF{T}_3^2)_i&=&e_0w_j\otimes [z^j,w_i]\otimes e_1+e_0w_i\otimes w\cdotp z\otimes e_1
-e_0\otimes w\cdotp z\otimes w_ie_1-e_0\otimes [w_i,w_j]\otimes z^je_1\nn\\
&\to&-\ep_{rik}\bar D^{rjl}[w_jw_lz^k]_{01}-\ep_{ijr}\bar C^{rkl}[w_kw_lz^j]_{01},\nn\eea
and we get $\bar C=-\bar D$.

The most laborious one is the ${\bf 8}={\tiny\yng(2,1)}$ component. To state the issue, we take for example $[w_i,z^j]_{tf}\in\SF{T}_2$, it will evaluate to $T_{ik}^{jl}w_lz^k$, for some tensor $T$ in the tensor product ${\bf 8}\otimes{\bf 8}$ (since the trace over $i,j$ and $k,l$ is zero). Not all such tensor is needed, we want it to correspond to the representation ${\bf 8}$. Similarly $z^{[i}z^{j]}$ will evaluate to $T^{ij}_{kl}z^kz^l$, we will also need $T$ to be in ${\bf 8}$. So we state a lemma to this effect
\begin{lemma}
  Let $X_i^j$ be traceless i.e. it lives in ${\bf 8}$, then
  \bea &&X\mapsto \frac{1}{20}X^{(j}_{(i}\gd^{k)}_{l)}-\frac14X^{[j}_{[i}\gd^{k]}_{l]}\label{used_IX}\\
  &&X\mapsto -\frac{1}{12}X^{(j}_{[i}\gd^{k)}_{l]}{\color{black}+}\frac{1}{12}X^{[j}_{(i}\gd^{k]}_{l)}\label{used_X}\eea
  are two embeddings of ${\bf 8}$ into ${\bf 8}\otimes{\bf 8}$, while
  \bea X\mapsto \frac{1}{3} X^{(j}_{[i}\gd^{k)}_{l]},~~X\mapsto \frac13X^{[j}_{(i}\gd^{k]}_{l)}\nn\eea
  are the embeddings to ${\bf 6}\otimes {\bf 3}$ and ${\bf 6}^*\otimes {\bf 3}^*$ respectively.
\end{lemma}
\begin{proof}
  Note there is one copy of ${\bf 8}$ in ${\bf 8}\otimes^s{\bf 8}$ and in ${\bf 8}\otimes^a{\bf 8}$ respectively, where $s,a$ denotes the symmetric or anti-symmetric tensor product.
  With indices written explicitly, then $T_{il}^{jk}$ satisfies $T_{pl}^{pk}=T_{ip}^{jp}=0$. The projection to the two copies of ${\bf 8}$ goes as
  \bea T\stackrel{pr_s}{\to} T^{jp}_{pi}+T^{pj}_{ip},~~~T\stackrel{pr_a}{\to} T^{jp}_{pi}- T^{pj}_{ip}.\nn\eea
  Take now $T$ to be the rhs of \eqref{used_IX}, then one can check that indeed $T_{pl}^{pk}=T_{ip}^{jp}=0$, and $T$ vanishes under $pr_a$ but projects back to $X$ under $pr_s$. One can check \eqref{used_X} in the same way.

  For the embedding to ${\bf 6}\otimes {\bf 3}$, the corresponding tensor $T_{il}^{jk}$ need be symmetric in $j,k$ and anti-symmetric in $i,l$, besides $T_{pq}^{(jk}\ep^{i)pq}=0$. One can check that $X^{(j}_{[i}\gd^{k)}_{l]}$ fulfills these properties. The coefficients such as $1/3$ are not important, we chose them so the embeddings are isometric.
\end{proof}
Now we can continue with treating ${\bf 8}$ in the 2-cochain. We set the 2-cochain to be
\bea &&[z\cdotp w]_{11}\to U^i_j[w_iz^j]_{11},~~~[w\cdotp z]_{11}\to V^i_j[w_iz^j]_{11},\nn\\
&&\left[z\cdotp w\right]_{22}\to X^i_j[z^jw_i]_{22},~~~[w\cdotp z]_{00}\to Y^i_j[w_iz^j]_{00},\nn\\
&&\left[w_i,z^j\right]_{tf,11}\to \big(\frac{1}{20}S^{(j}_{(i}\gd^{k)}_{l)}-\frac{1}{4}S^{[j}_{[i}\gd^{k]}_{l]}-\frac{1}{12}A^{(j}_{[i}\gd^{k)}_{l]}
{\color{black}+}\frac{1}{12}A^{[j}_{(i}\gd^{k]}_{l)}\big)[w_kz^l]_{11}\nn\\
&&\hspace{1.9cm}=(\frac{3}{10}S-\frac{1}{6}A)^k_i[w_kz^j]_{11}+(\frac{3}{10}S+\frac{1}{6}A)^j_l[w_iz^l]_{11}-\frac15S^k_l\gd^j_i[w_kz^l]_{11}\nn\\
&&\left[z^i,z^j\right]_{20}\to \frac13P^{[i}_{(k}\gd^{j]}_{l)}[z^kz^l]_{20},~~~\left[w_i,w_j\right]_{02}\to \frac13Q^{(k}_{[i}\gd^{l)}_{j]}[w_kw_l]_{02}\nn\eea
where $U,V,X,Y,S,A,P,Q$ are all traceless, in particular, for $S,A,P,Q$ we used the last lemma. Again evaluating $\partial (\SF{T}_3^1)^j$ on the above cochains gives
\bea \partial (\SF{T}_3^1)^j&\to& \big(\frac{1}{20}S^{(j}_{(i}\gd^{k)}_{l)}-\hcancel[blue]{\frac{1}{4}S^{[j}_{[i}\gd^{k]}_{l]}-\frac{1}{12}A^{(j}_{[i}\gd^{k)}_{l]}}
{\color{black}+}\frac{1}{12}A^{[j}_{(i}\gd^{k]}_{l)}\big)[z^iw_kz^l]_{21}\nn\\
&&+\frac13V^k_l[z^jw_kz^l]_{21}+\frac23U^k_l[z^jw_kz^l]
-X^k_l[z^lw_kz^j]_{21}+\frac13P^{[i}_{(k}\gd^{j]}_{l)}[z^kz^lw_i]_{21}\nn\\
&=&(\frac{1}{10}S^k_i-\frac16A^k_i+\frac{1}{3}V^k_i+\frac23U^k_i-X_i^k+\frac23P^k_i)[z^iw_kz^j]_{21}.\nn\eea
That $\partial (\SF{T}_3^1)^j$ should evaluate to zero gives us
\bea \frac{1}{10}S-\frac16A+\frac{1}{3}V+\frac23U-X+\frac23P=0.\label{Z1}\eea
Writing down the evaluation of all four $\partial \SF{T}_3$ is too painful, we just record the equations
\bea \frac13U+\frac23V-Y-(\frac{1}{10}S+\frac16A)+\frac23Q=0,\label{Z2}\\
-\frac13U-\frac23V+Y+(\frac{1}{10}S-\frac16A)+\frac23P=0,\label{Z3}\\
-\frac23U-\frac13V+X-(\frac{1}{10}S+\frac16A)+\frac23Q=0.\label{Z4}\eea
These four equations imply
\bea 0=-\frac13A+\frac23(P+Q)=U+V-X-Y.\label{cocycle_condition}\eea
Some of these components will be removed as coboundaries.

To investigate the coboundaries, we consider the 1-cochain
\bea [z^i]_{10}\to T^i_j[z^j]_{10}\nn\eea
and the resulting 2-coboundary, e.g.
\bea &&\partial[z\cdotp w]_{11}=\hcancel[blue]{e_1z^i\otimes w_i\otimes e_1}+e_1\otimes z^i\otimes w_ie_1 \to T^i_j[z^jw_i]_{11},\nn\\
&&\partial[z^{[i}z^{j]}]_{20}=e_2z^{[i}\otimes z^{j]}\otimes e_0+\hcancel[blue]{e_2\otimes z^{[i}\otimes z^{j]}e_0}\to T^{[j}_k[z^{i]}z^k]_{20},\nn\\
&&\partial[w_i,z^j]_{11}=\hcancel[blue]{e_1w_i\otimes z^j\otimes e_1-e_1z^j\otimes w_i\otimes e_1+e_1\otimes w_i\otimes z_je_1}-e_1\otimes z^j\otimes w_ie_1\to-T^j_k[z^kw_i]_{11}.\nn\eea
The last term needs to be decomposed into trace part and trace free part.
Thus this particular 1-cochain results in a coboundary
\bea (U,Y,S,A,P)=(T_{tf},T_{tf},-\frac53T_{tf},-3T_{tf},-\frac32T_{tf})\in{\bf 8}\label{B1}\eea
and the trace part of $T$ gives the coboundary
\bea [w_i,z^j]_{11}\to \partial[w_i,z^j]_{11}\to-[z^jw_i]_{11}\nn\eea
which cancels the singlet in \eqref{used_III}. At this stage, it is helpful to check that the cobounary \eqref{B1} is closed, by plugging it into Eqs.\ref{Z1}, \ref{Z2}, \ref{Z3} and \ref{Z4}, to make sure all coefficients are correct.

We list here all the 2-coboundaries
\bea (V,X,S,A,P)=(T_{tf},T_{tf},\frac53T_{tf},3T_{tf},\frac32T_{tf}),\label{B2}\\
  (U,Y,S,A,Q)=(T_{tf},T_{tf},-\frac53T_{tf},3T_{tf},\frac32T_{tf})\label{B3}\\
   (V,X,S,A,Q)=(T_{tf},T_{tf},\frac53T_{tf},-3T_{tf},-\frac32T_{tf}).\label{B4}\eea
from $[z^i]_{21}\to T^i_j[z^j]_{21}$, $[w_i]_{01}\to T_i^j[w_j]_{01}$ and $[w_i]_{01}\to T_i^j[w_j]_{12}$ respectively.

Now we can finally finish the calculation for $\HH^2|_{u^0}$. Note that the coboundaries cannot be in ${\bf 10}$ or ${\bf 10}^*$, thus the cocycles in these representations survive to cohomology. While for ${\bf 8}$ representation,
we can use \eqref{B1} to set $P=0$, then use $\eqref{B1}+\eqref{B2}$ to set $U+V+X+Y=0$. From \eqref{cocycle_condition}, we get $U+V=X+Y=0$.
Similarly, we use $\eqref{B3}-\eqref{B4}$ to set $Q=0$, then we get $A=0$ and $S=10X-10U/3$. Thus the cohomology in ${\bf 8}$ is parametrised by $U,X$ alone
and we get the explicit cocycle
\bea &&[z\cdotp w]_{11}\to U^i_j[w_iz^j]_{11},~~~[w\cdotp z]_{11}\to -U^i_j[w_iz^j]_{11},\nn\\
&&\left[z\cdotp w\right]_{22}\to X^i_j[z^jw_i]_{22},~~~[w\cdotp z]_{00}\to -X^i_j[w_iz^j]_{00},\nn\\
&&\left[w_i,z^j\right]_{tf,11}\to (3X-U)^k_i[w_kz^j]_{11}+(3X-U)^j_l[w_iz^l]_{11}-(2X-\frac23U)^k_l\gd^j_i[w_kz^l]_{11}.\nn\eea
These are the two components in ${\bf 8}$. There are also the components in ${\bf 10}$ and ${\bf 10}^*$ that we obtained earlier
\bea
{\bf 10}:&&[z^{[i}z^{j]}]_{20}\to \ep^{ijr}D_{rkl}[z^kz^l]_{20},~~~
\left[w_i,z^j\right]_{tf,11} \to \ep^{jlr}D_{rik}[w_lz^k]_{11},\nn\\
{\bf 10}^*:&&[w_{[i}w_{j]}]_{02}\to -\ep_{ijr}\bar D^{rkl}[w_kw_l]_{02},~~~~\left[w_i,z^j\right]_{tf,11} \to \ep_{rik}\bar D^{rjl}[w_lz^k]_{11}.\nn\eea


\end{document}